\documentclass[10pt,journal]{IEEEtran}

\ifCLASSINFOpdf

\else

\fi

\usepackage{color}
\usepackage{textcomp}
\usepackage[ruled]{algorithm2e}
\usepackage{graphicx}
\usepackage{subfigure}
\usepackage{amsmath,amssymb,amsfonts,mathrsfs,epsfig}
\usepackage{multirow}
\usepackage{amssymb}
\usepackage{algorithmic}
\usepackage{multirow}
\usepackage{changebar}
\usepackage{booktabs}
\usepackage{stfloats}
\usepackage{epstopdf}
\usepackage{xcolor,stfloats}
\usepackage{bm}
\usepackage[colorlinks,citecolor=blue,linkcolor=blue,urlcolor=blue]{hyperref}
\usepackage{subeqnarray}
\usepackage{cases}
\usepackage{tcolorbox}

\newtheorem{theorem}{Theorem}
\newtheorem{remark}{Remark}
\newtheorem{example}{Example}

\newtheorem{proof}{Proof}

\newtheorem{corollary}{Corollary}

\begin{document}

\title{Artificial Noise Versus Artificial Noise Elimination: Redefining Scaling Laws of Physical Layer Security}

\author{Hong Niu,
        Tuo Wu,
        Xia Lei,
        Wanbin Tang,
        M\'{e}rouane Debbah,~\IEEEmembership{Fellow, IEEE,}\\
        H. Vincent Poor,~\IEEEmembership{Life Fellow, IEEE,}
        and Chau Yuen,~\IEEEmembership{Fellow, IEEE}


\thanks{

Hong Niu, Tuo Wu, and Chau Yuen are with the School of Electrical and Electronics Engineering, Nanyang Technological University, Singapore 639798 (E-mail: $\rm \{hong.niu,tuo.wu,chau.yuen\}@ntu.edu.sg$).

X. Lei and W. Tang are with the National Key Laboratory of Wireless Communications, University of Electronic Science and Technology of China, Chengdu 611731, China (E-mail: $\rm \{leixia,wbtang.\}@uestc.edu.cn$).

M. Debbah is with the Research Institute for Digital Future, Khalifa University, 127788 Abu Dhabi, United Arab Emirates (E-mail: $\rm merouane.debbah@ku.ac.ae$).

H. V. Poor is with the Department of Electrical and Computer Engineering, Princeton University, Princeton, NJ, USA 08544 (E-mail: $\rm poor@princeton.edu$).

}
}


\markboth{}%
{Shell \MakeLowercase{\textit{et al.}}: }
\maketitle

\begin{abstract}
Artificial noise (AN) is a key physical-layer security scheme for wireless communications over multiple-input multiple-output wiretap channels. Recently, artificial noise elimination (ANE) has emerged as a strategy to mitigate the impact of AN on eavesdroppers. However, the influence of ANE on the secrecy rate when counteracting AN has not been investigated. In this paper, we address this issue by establishing scaling laws for both average and instantaneous secrecy rates in the presence of AN and ANE. Based on the scaling laws, several derived corollaries provide insights into the mutual constraints between the number of transmit antennas, receive antennas, and antennas at eavesdroppers, revealing the interplay between these factors. A key corollary reveals that when the eavesdropper possesses more than twice as many antennas as the transmitter, secure communication may no longer be guaranteed. Additionally, by comparing scenarios where ANE counteracts AN with those where AN is not employed, this study identifies sufficient conditions under which AN remains effective. Finally, the derived secrecy rates provide guidelines for system design, even in the presence of advanced ANE countermeasures implemented by the eavesdropper.
\end{abstract}


\IEEEpeerreviewmaketitle

\section{Introduction}

\IEEEPARstart{T}{he} security of data transmission has always been a critical concern in wireless communications, primarily due to the broadcast nature of signal propagation. Although conventional cryptographic mechanisms \cite{CM1}-\cite{CM3} are effective in enhancing secrecy, they entail extra computational overhead and increased latency, arising from the redundant processes of key distribution and management. More critically, these encryption methods rely on the assumption that decryption is infeasible without knowledge of the key, an assumption that becomes increasingly vulnerable with rapid advances in computational power. Moreover, the infrastructure needed for management of such methods is not necessarily present in all wireless systems.

In recent decades, physical-layer security (PLS) \cite{PLS1}-\cite{PLS5} has gained widespread recognition as a promising approach for achieving information-theoretic security in wireless communications. Unlike conventional cryptographic methods, PLS techniques leverage the inherent characteristics of wireless channels, such as noise, fading, and channel diversity, to enhance security at the physical layer. By utilizing the spatial and temporal variability of wireless channels, PLS techniques degrade the signal quality received by an unauthorized eavesdropper (Eve) while preserving the quality of the signal for the legitimate receiver (Bob), thereby improving security. The concept of PLS was first introduced by Wyner in \cite{wyner1}, who characterized the secrecy capacity as the difference in channel capacity between Bob and Eve, showing that perfect secrecy can be achieved when the secrecy capacity is positive.

Building on this principle, artificial noise (AN) \cite{AN1}-\cite{AN6} has emerged as a prominent PLS technique, utilizing legitimate channel state information (CSI) to degrade Eve's ability to intercept transmitted messages. Specifically, AN is generated by the legitimate transmitter (Alice) to exploit the spatial degrees-of-freedom (DoFs) of the wireless channel. By injecting AN into the null space of legitimate channels, it ensures that the interference affects Eve's reception without impacting Bob's. This strategy effectively preserves Bob's channel capacity while impairing Eve's, thus improving secrecy performance \cite{AN7}-\cite{AN16}. Given its effectiveness in securing wireless communications against eavesdropping, AN has been widely studied as a means to enhance the reliability of transmission in multiple-input multiple-output (MIMO) systems \cite{AN17}-\cite{AN22}.

Consequently, numerous scaling laws for applying AN have been derived, demonstrating how different parameters influence secrecy performance. For instance, studies have shown that more power should be allocated to the AN when Eve has more antennas \cite{AN8,AN15}. Moreover, the authors of \cite{ANE2} have derived closed-form expressions for the average and instantaneous secrecy rates under the AN framework.

However, as AN methods have advanced, countermeasures have also evolved to mitigate its impact, exposing vulnerabilities in PLS. Specifically, the emergence of artificial noise elimination (ANE) techniques has challenged the effectiveness of AN-based schemes \cite{ANE1}-\cite{HC}. For instance, the zero-forcing elimination scheme proposed in \cite{ANE1} allows an Eve equipped with more antennas than Alice to replicate the Alice-Bob channel, effectively mitigating the impact of AN. Similarly, the null-space elimination scheme introduced in \cite{ANE3} uses selective projections to enhance detection quality, even with fewer antennas. Furthermore, advanced methods such as the principal component analysis-based ANE scheme \cite{Fisher1} and the hyperplane clustering algorithm \cite{HC} have demonstrated the capability to suppress the impact of AN without requiring legitimate CSI.

These developments underscore the growing physical-layer insecurity \cite{PLI1,PLI2} caused by sophisticated eavesdropping techniques, raising an important question: \textit{Do existing AN scaling laws adequately capture the challenges posed by such advanced countermeasures?} This question motivates a \textbf{\textit{re-definition of scaling laws when AN is countered by ANE }} to better understand this threat to PLS.

Against this background, this paper characterizes both average and instantaneous secrecy rates in the presence of AN and ANE.
The key contributions of this paper are threefold.

\begin{enumerate}
  \item \textbf{Redefinition of Scaling Laws}: We categorize the scenarios into two groups based on whether or not ANE can fully eliminate the impact of AN. In the case where AN is completely mitigated by ANE, we derive closed-form expressions for both the average secrecy rate and the normalized instantaneous secrecy rate. In the case where residual AN persists, we provide a semi-analytical expression for the average secrecy rate and a closed-form expression for the normalized instantaneous secrecy rate.

  \item \textbf{Insightful Corollaries}: Building upon the redefined scaling laws, we derive several straightforward yet significant corollaries. Specifically, to determine the conditions under which Eve can achieve perfect eavesdropping, we provide sufficient conditions for the secrecy rates to become zero. It is shown that the ANE scheme enables Eve to achieve perfect eavesdropping despite the use of an AN scheme, when the number of antennas at the eavesdropper satisfies the condition ${N_e} \geqslant 2{N_a}-{N_b}$, where ${N_a}$, ${N_b}$, and ${N_e}$ denote the numbers of antennas at Alice, Bob, and Eve, respectively. Furthermore, we approximate the average and instantaneous secrecy rates and present a sufficient condition for achieving the secrecy capacity using Gaussian input alphabets.

  \item \textbf{Re-examination of the Effectiveness of AN}: To fairly evaluate the benefits of introducing AN, we derive closed-form expressions for both the average secrecy rate and the normalized instantaneous secrecy rate in the absence of AN. The derived corollaries reveal that, without AN, Eve can achieve perfect eavesdropping when the number of her antennas satisfies ${N_e} \geqslant {N_b}$. Furthermore, we provide sufficient conditions under which the introduction of AN remains effective in enhancing PLS.

\end{enumerate}

This paper is organized as follows. Section II presents the system model under AN and ANE. In Section III, new scaling laws for AN versus ANE are derived, including expressions for both the average secrecy rate and the normalized instantaneous secrecy rate. Section IV provides several key corollaries to the new scaling laws. The effectiveness of AN is discussed in Section V. Finally, conclusions are drawn in Section VI.

\textit{Notation}: In the following, bold lower-case and upper-case letters denote vectors and matrices, respectively. ${\left(  \cdot  \right)^H}$, $\left|  \cdot  \right|$ and $\left\|  \cdot  \right\|$ represent the conjugate transpose, modulus and norm operation, respectively. ${\left[ x \right]^ + } = \max \left( {x,0} \right)$ returns $x$ when $x$ is positive, and
0 otherwise. $\mathcal{C}\mathcal{N}({\mu },\sigma ^2)$ represents the complex Gaussian distribution with mean ${\mu }$ and variance $\sigma ^2$. The complex number field and the order of computational complexity are denoted by $\mathbb{C}$ and $\mathcal{O}$, respectively. $\mathcal{C}\mathcal{N}({\mathbf{u }},{\mathbf{\Sigma }})$ represents the distribution of a circularly symmetric complex Gaussian random vector with mean vector ${\mathbf{u }}$ and covariance variance ${\mathbf{\Sigma }}\succeq{\mathbf{0}}$.

\section{System model}

\begin{figure}[t]
\centering
\includegraphics[width=3.5in]{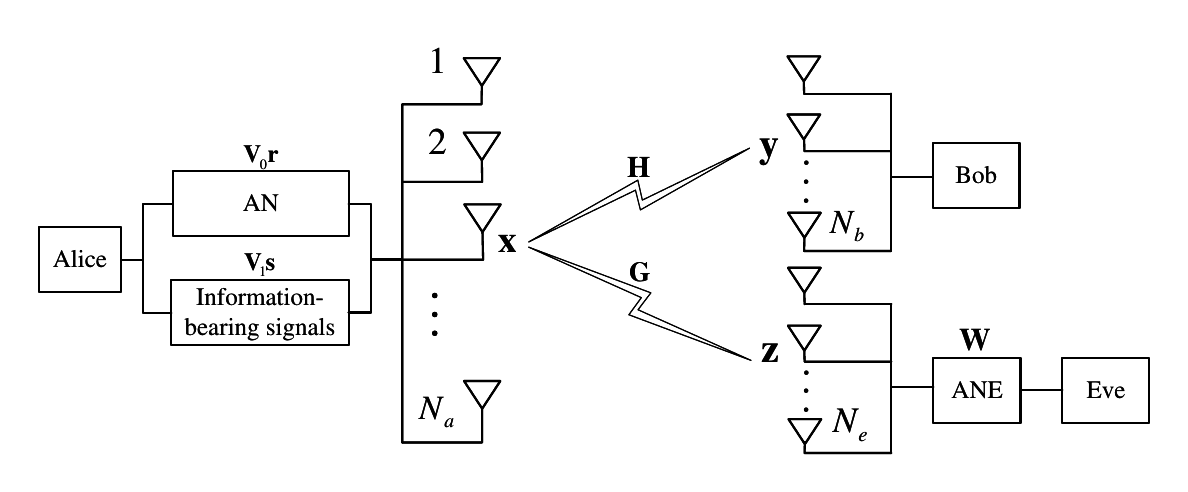}
\vspace{+3mm}
\caption{System model of AN vs. ANE in MIMO wireless communications.}
\label{fig_1}
\vspace{-0em}
\end{figure}

We consider a MIMO wireless communication system, where Alice communicates with Bob in the presence of a passive Eve. As illustrated in Fig. \ref{fig_1}, Alice, Bob, and Eve are equipped with ${N_a}$, ${N_b}$, and ${N_e}$ antennas, respectively. The received signals at Bob and Eve can be expressed as
\begin{equation}
{\mathbf{y}} = {\mathbf{Hx}} + {\mathbf{u}},
\end{equation}
\begin{equation}
{\mathbf{z}} = {\mathbf{Gx}} + {\mathbf{n}},
\end{equation}
where ${\mathbf{x}} \in {\mathbb{C}^{{N_a}}}$ denotes the transmitted signal at Alice, ${\mathbf{H}} \in {\mathbb{C}^{{N_b} \times {N_a}}}$ and ${\mathbf{G}} \in {\mathbb{C}^{{N_e} \times {N_a}}}$ represent the legitimate channel from Alice to Bob and the wiretap channel from Alice to Eve, respectively, while ${\mathbf{u}} \in {\mathbb{C}^{{N_b}}}$ and ${\mathbf{n}} \in {\mathbb{C}^{{N_e}}}$ are complex additive white Gaussian noise (AWGN) vectors with distributions $\mathcal{C}\mathcal{N}({\mathbf{0}},\sigma _u^2{{\mathbf{I}}_{{N_b}}})$ and $\mathcal{C}\mathcal{N}({\mathbf{0}},\sigma _n^2{{\mathbf{I}}_{{N_e}}})$, respectively, i.e., $\sigma _u^2$ and $\sigma _n^2$ represent the noise power at Bob and Eve, respectively. In this paper, we assume that the legitimate and wiretap channels are uncorrelated MIMO Rayleigh-fading channels, i.e., the entries of ${\mathbf{H}}$ and ${\mathbf{G}}$ are i.i.d. $\mathcal{C}\mathcal{N}(0,1)$.

Additionally, we assume AN is employed by Alice to interfere with Eve, with the goal of securing the legitimate transmission to Bob, while ANE is employed at Eve to reduce the impact of the AN.

\subsection{AN Scheme}

To enhance the security of the system, the information-bearing signal and AN interference are combined in the transmitted signal as
\begin{equation}\label{x:1}
{\mathbf{x}} = {{\mathbf{V}}_1}{\mathbf{s}} + {{\mathbf{V}}_0}{\mathbf{r}} = {\mathbf{V}}\left[ {\begin{array}{*{20}{c}}
  {\mathbf{s}} \\
  {\mathbf{r}}
\end{array}} \right],
\end{equation}
where ${\mathbf{s}} \in {\mathbb{C}^{{N_b}}}$ denotes the information-bearing signal and ${\mathbf{r}} \in {\mathbb{C}^{{N_a} - {N_b}}}$ represents a random vector to generate AN. ${{\mathbf{V}}_1} \in {\mathbb{C}^{{N_a} \times {N_b}}}$ spans the characteristic subspace corresponding to the effective information-bearing signal part of the transmission, and ${{\mathbf{V}}_0} \in {\mathbb{C}^{{N_a} \times \left( {{N_a} - {N_b}} \right)}}$ spans the null space associated with the noisy part of the transmission. Their combination, ${\mathbf{V}} = \left[ {\begin{array}{*{20}{c}}
  {{{\mathbf{V}}_1}}&{{{\mathbf{V}}_0}}
\end{array}} \right]$, forms a unitary matrix containing right singular vectors of ${{\mathbf{H}}}$, which can be obtained by applying a singular value decomposition (SVD) to ${\mathbf{H}}$ as
\begin{equation}\label{svd:1}
{\mathbf{H}} = {\mathbf{U}}\left[ {\begin{array}{*{20}{c}}
  {\mathbf{\Lambda }}&{{{\mathbf{0}}_{{N_a} - {N_b}}}}
\end{array}} \right]{{\mathbf{V}}^H},
\end{equation}
where ${\mathbf{U}} \in {\mathbb{C}^{{N_b} \times {N_b}}}$ denotes a unitary matrix consisting of left singular vectors of ${{\mathbf{H}}}$ and ${\mathbf{\Lambda }} \in {\mathbb{C}^{{N_b} \times {N_b}}}$ represents a diagonal matrix whose diagonal elements are singular values of ${{\mathbf{H}}}$.

{After the transmitted signal passes through the legitimate channel ${\mathbf{H}}$ and the wiretap channel ${\mathbf{G}}$, the received signals at Bob and Eve can be expressed as
\begin{equation}\label{A:2}
{\mathbf{y}} = {\mathbf{H}}{{\mathbf{V}}_1}{\mathbf{s}} + {\mathbf{u}}, \ \ \ \ \ \ \ \ \ \ \
\end{equation}
\begin{equation}\label{A:1}
{{\mathbf{\tilde z}}} = {\mathbf{G}}{{\mathbf{V}}_1}{\mathbf{s}} + {\mathbf{G}}{{\mathbf{V}}_0}{\mathbf{r}} + {\mathbf{n}},
\end{equation}
where ${\mathbf{y}} \in {\mathbb{C}^{{N_b}}}$ denotes the received signal at Bob and ${\mathbf{\tilde z}} \in {\mathbb{C}^{{N_e}}}$ represents the received signal at Eve without applying ANE.}

As inferred from (\ref{A:2}) and (\ref{A:1}), the AN ${{\mathbf{V}}_0}{\mathbf{r}}$ has no impact on the reception of Bob due to the null space property ${\mathbf{H}}{{\mathbf{V}}_0} = {\mathbf{0}}$, while it significantly degrades that of Eve due to the difference between the legitimate channel ${\mathbf{H}}$ and the wiretap channel ${\mathbf{G}}$.

\subsection{ANE Scheme}

The efficacy of AN in enhancing PLS has been extensively studied, often under the implicit or explicit assumption that Eve lacks the capacity to mitigate its effects. This assumption supports the common belief that AN always degrade Eve's channel quality. However, recent studies have challenged this view by demonstrating that advanced eavesdroppers can implement ANE techniques to counteract AN interference \cite{ANE1}-\cite{HC}. Specifically, sophisticated Eves can leverage techniques, such as processing multiple received vectors ${\mathbf{z}}$, to estimate the effective AN channel, represented by the composite matrix ${\mathbf{G}}{{\mathbf{V}}_0}$. With knowledge or an estimate of this matrix, Eve can subsequently apply projections or other signal processing methods to significantly suppress, or in certain conditions completely nullify, the AN interference. This capability implies that previous security analyses overlooking the potential for ANE might yield overly optimistic assessments of the secrecy rate improvements afforded by AN.

{To provide a more rigorous analysis, this paper investigates the fundamental secrecy limits under the following scenario:

1) Alice is assumed to have knowledge of the legitimate channel ${\mathbf{H}}$, enabling her to generate AN confined to Bob's null space, i.e., ${\mathbf{H}}{{\mathbf{V}}_0} = {\mathbf{0}}$;

2) Eve is assumed to have knowledge of the composite channel ${\mathbf{G}}{{\mathbf{V}}_0}$, empowering her to implement informed and potentially highly effective ANE strategies.

Note that Eve cannot acquire Bob's channel ${\mathbf{H}}$ and hence cannot directly obtain ${{\mathbf{V}}_0}$. Nevertheless, Eve may estimate the composite channel ${\mathbf{G}}{{\mathbf{V}}_0}$ by jointly processing multiple received signal vectors over a channel coherence interval. This assumption represents an idealized scenario for secrecy analysis and yields a conservative secrecy bound. In practical systems, estimation errors are inevitable, and their impact is left for future investigation.}

This scenario enables Eve to employ the following ANE countermeasure as
\begin{equation}\label{W:1}
\begin{gathered}
  \mathop {\min }\limits_{\mathbf{W}} {\text{tr}}\left( {{\mathbf{WG}}{{\mathbf{V}}_0}{\mathbf{V}}_0^H{{\mathbf{G}}^H}{{\mathbf{W}}^H}} \right) \hfill \\
  {\text{s}}{\text{.t}}{\text{. tr}}\left( {{\mathbf{W}}{{\mathbf{W}}^H}} \right) = {N_w}, \hfill \\
\end{gathered}
\end{equation}
where ${\mathbf{W}} \in {\mathbb{C}^{{N_w} \times {N_e}}}$ denotes the projection matrix for reducing the impact of AN and ${{N_w}}$ represents the dimension of the signal after projection. {In the scenario of ${N_e} > {N_a} - {N_b}$, the optimal projection matrix ${\mathbf{W}}$ is given by (\ref{optW:1}) with ${N_w} = {N_a} - {N_b} + 1$. Conversely, in the scenario of ${N_e} \leqslant {N_a} - {N_b}$, the optimal projection reduces to a vector ${{\mathbf{w}}^H}$ given in (\ref{optw:2}) with ${N_w} = 1$.}

{When Eve employs the ANE countermeasure, the processed received signal can be rewritten as
\begin{equation}
{\mathbf{z}} = {\mathbf{WG}}{{\mathbf{V}}_1}{\mathbf{s}} + {\mathbf{WG}}{{\mathbf{V}}_0}{\mathbf{r}} + {\mathbf{Wv}}.
\end{equation}}

\subsection{Average and Instantaneous Secrecy Rates}

In this paper, we use the following definitions of the average secrecy capacity ${\bar C_s}$, the average secrecy rate ${{\bar R}_s}$, and the instantaneous secrecy rate ${R_s}$:
\begin{equation}\label{cs:1}
{\bar C_s} \triangleq \mathop {\max }\limits_{p\left( {\mathbf{s}} \right)} {\bar R_s}, \ \ \ \ \ \ \ \ \ \ \ \ \ \ \ \ \ \ \ \ \ \ \ \ \ \ \ \ \ \ \ \ \ \ \ \ \ \ \ \ \ \
\end{equation}
\begin{equation}\label{vr:2}
\begin{gathered}
  {{\bar R}_s} \triangleq {\left[ {{{\bar R}_b} - {{\bar R}_e}} \right]^ + } \hfill \\
  \ \ \ \  = {\left( {{{\text{E}}_{\mathbf{H}}}\left[ {I\left( {{\mathbf{s}};{\mathbf{y}}\left| {\mathbf{H}} \right.} \right)} \right] - {{\text{E}}_{{\mathbf{H}},{\mathbf{G}}}}\left[ {I\left( {{\mathbf{s}};{\mathbf{z}}\left| {{\mathbf{H}},{\mathbf{G}}} \right.} \right)} \right]} \right)^ + }, \hfill \\
\end{gathered}
\end{equation}
\begin{equation}\label{vr:1}
{R_s} \triangleq {\left[ {{R_b} - {R_e}} \right]^ + } = {\left[ {I\left( {{\mathbf{s}};{\mathbf{y}}} \right) - I\left( {{\mathbf{s}};{\mathbf{z}}} \right)} \right]^ + }, \ \ \ \ \ \ \ \ \
\end{equation}
where {$I\left( {{\mathbf{s}};{\mathbf{y}}} \right)$ and $I\left( {{\mathbf{s}};{\mathbf{z}}} \right)$ denote the mutual information between the information-bearing signal ${\mathbf{s}}$ and the received signal vectors ${\mathbf{y}}$ at Bob and ${\mathbf{z}}$ at Eve, respectively, quantifying how much information about ${\mathbf{s}}$ can be extracted from
${\mathbf{y}}$ and ${\mathbf{z}}$}. In addition, ${{\rm E}_Z}\left[ {I\left( {X;Y} \right)\left| Z \right.} \right]$ represents the average mutual information (see the definitions of (12), (13), (14) in \cite{ANE2}).

\section{Redefinition of Scaling Laws for AN}

In this section, we present a closed-form expression for the average secrecy rate in (\ref{vr:2}). Then we derive an asymptotic analysis on the instantaneous secrecy rate in (\ref{vr:1}). To simplify the presentation of the derivation results, we define several parameters and functions.

\subsection{Definitions}

To simplify the notation, we define three system parameters: the SNR of Eve, the power ratio of AN to the information-bearing signal, and the Eve-to-Bob noise power ratio, as
\begin{equation}
\alpha  \triangleq \sigma _s^2/\sigma _n^2,
\end{equation}
\begin{equation}
\beta  \triangleq \sigma _r^2/\sigma _s^2,
\end{equation}
and
\begin{equation}
\gamma  \triangleq \sigma _n^2/\sigma _u^2,
\end{equation}
respectively.

Note that the SNR of Bob is given by $\alpha \gamma $, and the power of the information-bearing signal and AN can be expressed as
\begin{equation}
{P_s} = {\left\| {\mathbf{s}} \right\|^2} = \alpha \gamma {N_b}\sigma _u^2, \ \ \ \ \ \ \ \ \ \ \ \
\end{equation}
\begin{equation}
{P_r} = {\left\| {\mathbf{r}} \right\|^2} = \alpha \beta \gamma \left( {{N_a} - {N_b}} \right)\sigma _u^2.
\end{equation}

Next, we give a closed-form ergodic capacity integrand for the Rayleigh MIMO channel (see (8) in \cite{SR1}) as
\begin{equation}\label{func:1}
\begin{small}
\begin{gathered}
  f\left( {t,r,x} \right) \triangleq {e^{ - 1/x}}{\log _2}\left( e \right)\sum\limits_{k = 0}^{m - 1} {\sum\limits_{l = 0}^k {\sum\limits_{i = 0}^{2l} {\left\{ {\sum\limits_{j = 0}^{n - m + i} {{E_{j + 1}}\left( {1/x} \right)} } \right.} } }  \hfill \\
  \left. {\frac{{{{\left( { - 1} \right)}^i}\left( {2l} \right)!\left( {n - m + l} \right)!}}{{{2^{2k - 1}}l!i!\left( {n - m + l} \right)!}}\left( {\begin{array}{*{20}{c}}
  {2\left( {k - l} \right)} \\
  {k - l}
\end{array}} \right)\left( {\begin{array}{*{20}{c}}
  {2\left( {l + n - m} \right)} \\
  {2l - i}
\end{array}} \right)} \right\}, \hfill \\
\end{gathered}
\end{small}
\end{equation}
where $m$ and $n$ define two values based on the minimum and maximum of $t$ and $r$ as
\begin{equation}
m = \min \left\{ {t,r} \right\},n = \max \left\{ {t,r} \right\},
\end{equation}
and ${E_p}\left( z \right)$ is the exponential integral of order $p$ (see (9) in \cite{SR1}), given by
\begin{equation}
{E_p}\left( z \right) = \int_1^\infty  {{e^{ - zx}}{x^{ - p}}{\text{d}}x} .
\end{equation}

Additionally, we will use an expression of the incomplete Gamma function, shown as
\begin{equation}\label{z:1}
\Gamma \left( {a,b} \right) = \int_b^\infty  {{x^{a - 1}}{e^{ - x}}{\text{d}}x} .
\end{equation}

\subsection{Average Secrecy Rate}

The analysis of the average secrecy rate can be divided into two scenarios, i.e., ${N_e} > {N_a} - {N_b}$ and ${N_e} \leqslant {N_a} - {N_b}$. In the former scenario, the ANE scheme can perfectly eliminate the impact of AN at Eve, while in the latter scenario, using ANE only partially mitigates the impact of AN and may leave residual AN power. The presence of residual AN affects the expression for the achievable rate at Eve, thereby influencing the average secrecy rate. The key distinction lies in the mathematical formulation of Eve's average achievable rate, $\bar{R}_e$, in these two cases.

\subsubsection{${N_e} > {N_a} - {N_b}$ (Complete AN Elimination)}

In this scenario, a closed-form expression for the average secrecy rate in (\ref{vr:2}) can be obtained from (56) in \cite{ANE2}, resulting in the following theorem.

\begin{theorem}
In the case of ${N_e} > {N_a} - {N_b}$, the average secrecy rate is expressed as
\begin{equation}\label{theo:1}
\begin{aligned}
  {{\bar R}_s} &= {\left[ {{{\bar R}_b} - {{\bar R}_e}} \right]^ + } \\
               &= {\left[ {f\left( {{N_a},{N_b},\alpha \gamma {N_a}} \right) - f\left( {{N_b},{N_e} - {N_a} + {N_b},\alpha {N_b}} \right)} \right]^ + }.
\end{aligned}
\end{equation}
\end{theorem}
\begin{proof}
See Appendix A.
\end{proof}

Theorem 1 provides a closed-form expression for the average secrecy rate in the regime $N_e > N_a - N_b$. Eve's average achievable rate in this case is given by
\begin{equation}\label{areve:1}
{\bar R}_e = f\left( {{N_b},{N_e} - {N_a} + {N_b},\alpha {N_b}} \right),
\end{equation}
which utilizes the function $f(\cdot)$ defined in (\ref{func:1}) for an equivalent interference-free MIMO channel.

{Although some analytical tools employed in this paper are related to existing results in \cite{ANE2}, the problem formulation, secrecy model, and theoretical conclusions are fundamentally different. In particular, prior works assume that AN always causes non-negligible interference at Eve, regardless of its antenna configuration or signal processing capability. In contrast, this paper considers a more powerful adversary capable of performing ANE. Under this model, the average achievable rate at Eve exhibits different behavior, especially in the regime ${N_e} > {N_a} - {N_b}$, where AN can be perfectly mitigated. As a result, the derived secrecy rate expressions, including Theorem 1, cannot be obtained by a direct application of existing results.}

\begin{figure}[t]
\centering
\includegraphics[width=3.5in,height=2.8in]{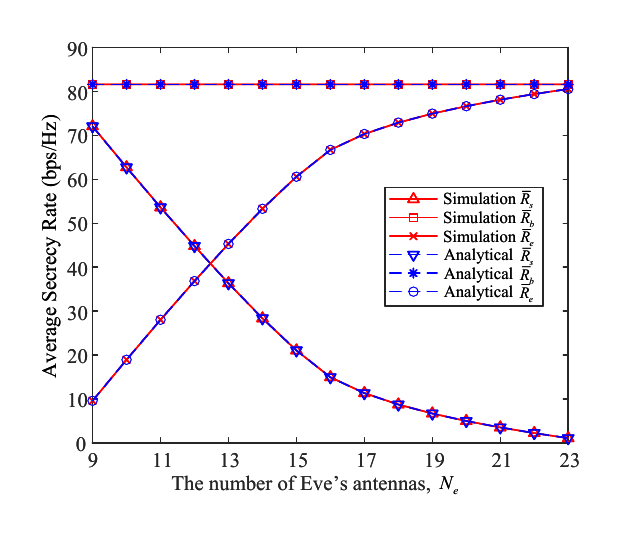}
\vspace{+3mm}
\caption{${{\bar R}_s}$, ${{\bar R}_b}$, and ${{\bar R}_e}$ versus $N_e$ with $\alpha  = 20$ dB, $\beta  = 1$, $\gamma  = 1$, $N_a=16$, and $N_b=8$.}
\label{fig_2}
\vspace{-0em}
\end{figure}

\begin{remark}
It can be observed that the AN power ratio $\beta$ is absent in (\ref{areve:1}), confirming that the ANE scheme completely eliminates the impact of AN in the case of ${N_e} > {N_a} - {N_b}$.
\end{remark}

\begin{example}
Fig. \ref{fig_2} plots the analytical average secrecy rate and its simulation results in an effort to validate the correctness of (\ref{theo:1}). The system configuration is set with parameters $\alpha  = 20$ dB, $\beta  = 1$, $\gamma  = 1$, $N_a=16$, and $N_b=8$, while $N_e$ varies from 9 to 23. As expected, the analytical average secrecy rate and achievable rates, i.e., ${{\bar R}_s}$, ${{\bar R}_b}$, and ${{\bar R}_e}$, align perfectly with the simulation results. As the number of antennas at the eavesdropper $N_e$ increases, the average achievable rate at Eve rises, leading to a corresponding decrease in the average secrecy rate.
\end{example}

\subsubsection{${N_e} \leqslant {N_a} - {N_b}$ (Residual AN Interference)}

When the number of Eve's antennas is insufficient, the impact of AN cannot be perfectly eliminated, and some residual AN power remains. In this case, the expression for the average secrecy rate in (\ref{vr:2}) is modified as follows.

\begin{theorem}
In the case of ${N_e} \leqslant {N_a} - {N_b}$, the average secrecy rate is given by
\begin{equation}\label{theo:2}
\begin{aligned}
  {{\bar R}_s} &= {\left[ {{{\bar R}_b} - {{\bar R}_e}} \right]^ + } \\
    &= \bigg[ f\left( {{N_a},{N_b},\alpha \gamma {N_a}} \right) - \\
    &\quad \int_0^\infty \int_0^\infty {\log_2 \left( {1 + \frac{{\alpha A}}{{\alpha \beta B + 1}}} \right){f_A}\left( a \right){f_B}\left( b \right){\text{d}a}{\text{d}b}} \bigg]^+,
\end{aligned}
\end{equation}
where
\begin{equation}
{f_A}\left( a \right) = \frac{{{a^{{N_b} - 1}}{e^{ - a}}}}{{\Gamma \left( {{N_b},0} \right)}}
\end{equation}
denotes the probability density distribution (PDF) of the scaled signal power $A \sim {\chi ^2}(2{N_b})/2$, and
\begin{equation}
{f_B}\left( b \right) = K\sum\limits_{p = 1}^{{N_e}} {\sum\limits_{q = 1}^{{N_e}} {{{\left( { - 1} \right)}^{p + q}}{b^{p + q - 2 + {N_a} - {N_b} - {N_e}}}{e^{ - b}}} } \left| {\mathbf{\Omega }} \right|
\end{equation}
represents the PDF of the squared minimum eigenvalue $B = \lambda_{\min}^2$ of the matrix ${{\mathbf{G}}{{\mathbf{V}}_0}{\mathbf{V}}_0^H{{\mathbf{G}}^H}}$, with
\begin{equation}
K = {\left[ {\prod _{i = 1}^{{N_e}}\left( {{N_a} - {N_b} - i} \right)!\prod _{j = 1}^{{N_e}}\left( {{N_e} - j} \right)!} \right]^{ - 1}},
\end{equation}
\begin{equation}
{\Omega _{i,j}} = \Gamma \left( {\alpha _{i,j}^{\left( p \right)\left( q \right)} + {N_a} - {N_b} - {N_e} + 1,b} \right),
\end{equation}
\begin{equation}
\alpha _{i,j}^{\left( p \right)\left( q \right)} = \left\{ \begin{gathered}
  i + j - 2,{\text{   if }}i < p{\text{ and }}j < q, \\
  i + j,{\text{ \ \ \ \ \       if }}i \geqslant p{\text{ and }}j \geqslant q, \\
  i + j - 1,{\text{    otherwise}}{\text{.}} \ \ \ \ \ \ \ \ \ \ \\
\end{gathered}  \right.
\end{equation}
\end{theorem}
\begin{proof}
See Appendix B.
\end{proof}

Theorem 2 shows that when $N_e \le N_a - N_b$, Eve's average achievable rate requires calculating an expectation via a double integral as
\begin{equation}\label{areve:2}
{\bar R}_e = \int_0^\infty \int_0^\infty {\log_2 \left( {1 + \frac{{\alpha A}}{{\alpha \beta B + 1}}} \right){f_A}\left( a \right){f_B}\left( b \right){{\text{d}}a}{{\text{d}}b}}.
\end{equation}
This semi-analytical formula averages the instantaneous rate over the distributions of the signal power component $A$ and the residual AN power component $B$. Unlike the complete elimination case, this expression for $\bar{R}_e$ explicitly depends on the AN power ratio $\beta$ through the SINR term $\alpha A/\left( {\alpha \beta B + 1} \right)$.

\begin{figure}[t]
\centering
\includegraphics[width=3.5in,height=2.8in]{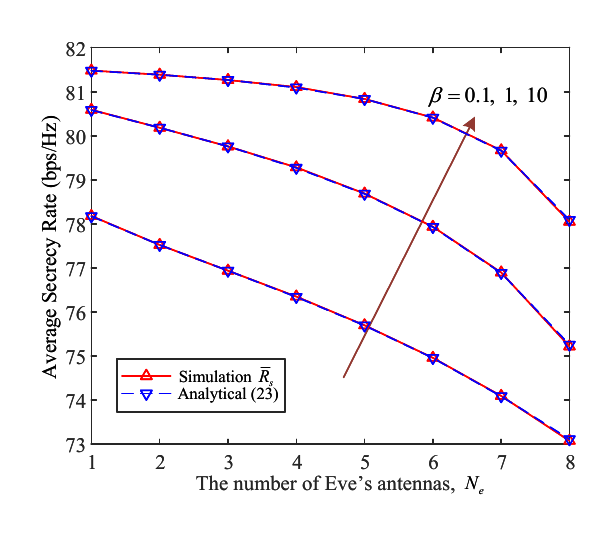}
\vspace{+3mm}
\caption{${{\bar R}_s}$ versus $N_e$ with $\alpha  = 20$ dB, $\gamma = 1$, $N_a=16$, and $N_b=8$.}
\label{fig_3}
\vspace{-0em}
\end{figure}

\begin{remark}
If Eve is equipped with only a single antenna ($N_e=1$), which naturally satisfied $N_e \le N_a - N_b$ due to the AN generation requiring ${N_a} > {N_b}$, the ANE scheme based on minimum eigenvalue projection offers no benefit, as the sole eigenvalue is the minimum eigenvalue. The AN power cannot be reduced further by projection in this specific subcase.
\end{remark}

\begin{example}
Fig. \ref{fig_3} depicts the analytical average secrecy rate from (\ref{theo:2}) alongside its simulation results for the case where ${{N_e}}  \leqslant  {{N_a}}-{{N_b}}$. The system configuration is set with $\alpha  = 20$ dB, $\gamma = 1$, $N_a=16$, and $N_b=8$, with ${{N_e}}$ varying from 1 to 8. As anticipated, the analytical average secrecy rate closely matches the simulation results. Furthermore, a larger value of $\beta$ leads to a higher secrecy rate, as the increased residual AN power further degrades Eve's channel.
\end{example}

\subsubsection{Comparison of Eve's Average Achievable Rates}

The threshold $N_e = N_a - N_b$ creates a fundamental dichotomy in the mathematical characterization of the average eavesdropper rate $\bar{R}_e$. This distinction provides critical insights into the relationship between antenna configuration and ANE effectiveness.

For $N_e > N_a - N_b$, Eve's average achievable rate takes the closed-form expression given in (\ref{areve:1}), which corresponds to the ergodic capacity of an equivalent interference-free MIMO channel with $N_b$ transmit dimensions and $N_e - N_a + N_b$ receive dimensions. This formulation is a direct consequence of Eve's ability to completely eliminate AN interference through appropriate projection. Most notably, $\bar{R}_e$ is independent of the AN power ratio $\beta$, confirming that when Eve has sufficient spatial DoFs, Alice's AN power becomes irrelevant to the system's security performance. This theoretical property is empirically validated in Fig. \ref{fig_2}, the average eavesdropper rate $\bar{R}_e$ steadily increases with $N_e$, following the closed-form expression in (\ref{areve:1}). In the regime $N_e > N_a - N_b$, increasing Eve's spatial DoFs directly translates to higher achievable rates, regardless of the AN power setting.

Conversely, for $N_e \leq N_a - N_b$, Eve's average achievable rate requires evaluation of the double integral expression in (\ref{areve:2}), which is a semi-analytical expression requiring numerical integration. This formulation arises because Eve can only partially mitigate the AN, leaving residual interference proportional to the minimum eigenvalue of ${\mathbf{G}}{{\mathbf{V}}_0}{\mathbf{V}}_0^H{{\mathbf{G}}^H}$. Unlike the complete elimination case, $\bar{R}_e$ explicitly depends on $\beta$ through the SINR term $\alpha A/\left( {\alpha \beta B + 1} \right)$, indicating that increasing AN power effectively degrades Eve's reception. This mathematical dependence on $\beta$ is clearly demonstrated in Fig. \ref{fig_3}, where increasing the AN power ratio from $\beta = 0.5$ to $\beta = 2$ results in progressively higher secrecy rates across the entire range of $N_e \leq N_a - N_b$. This confirms that when Eve lacks sufficient spatial DoFs to fully eliminate AN, the AN power allocation remains a critical design parameter for security.

The contrasting mathematical structures of $\bar{R}_e$ in these two regimes, empirically validated by Figs. \ref{fig_2} and \ref{fig_3}, demonstrate that the effectiveness of AN as a security mechanism is fundamentally determined by the relationship between Eve's spatial DoFs ($N_e$) and Alice's AN transmission dimensions ($N_a - N_b$). This insight provides clear design guidelines for secure MIMO systems operating in the presence of sophisticated eavesdroppers employing ANE techniques.

\subsection{Normalized Instantaneous Secrecy Rate}

The instantaneous secrecy rate $R_s$ in (\ref{vr:1}) depends on instantaneous channels ${\mathbf{H}}$ and ${\mathbf{G}}$. For robust secure communication design, it is crucial to characterize the asymptotic behavior of the normalized instantaneous secrecy rate $R_s/N_b$ as the number of antennas increases. We present analytical expressions that provide fundamental insights into how system security scales with antenna configurations.

\begin{theorem}
As ${N_a},{N_b},{N_e} \to \infty $ with ${N_e}/{N_b} \to {\delta _1}$, ${N_a}/{N_b} \to {\delta _2}$, and ${\delta _1} > {\delta _2} - 1$, the asymptotic behavior of the normalized instantaneous secrecy rate can be expressed as
\begin{equation}\label{theo:3}
\frac{{{R_s}}}{{{N_b}}} \to {\left[ {\Phi \left( {\alpha \gamma {N_b},{\delta _2}} \right) - \Phi \left( {\alpha {N_b},{\delta _1} - {\delta _2} + 1} \right)} \right]^ + },
\end{equation}
where $\Phi \left( {x,y} \right)$ is defined as
\begin{equation}\label{fidef:1}
\begin{aligned}
  \Phi \left( {x,y} \right) &= y\log \left( {1 + x - \frac{{\mathcal{F}\left( {x,y} \right)}}{4}} \right) + \\
   & \log \left( {1 + xy + \frac{{\mathcal{F}\left( {x,y} \right)}}{4}} \right) - \frac{{\log e}}{{4x}}\mathcal{F}\left( {x,y} \right), \\
\end{aligned}
\end{equation}
\begin{equation}\label{Fdef:1}
\begin{small}
\mathcal{F}\left( {x,y} \right) = {\left( {\sqrt {x{{\left( {1 + \sqrt y } \right)}^2} + 1}  - \sqrt {x{{\left( {1 - \sqrt y } \right)}^2} + 1} } \right)^2}.
\end{small}
\end{equation}
\end{theorem}

\begin{proof}
See Appendix C.
\end{proof}

{Theorem 3 reveals a fundamentally different physical interpretation from the results in \cite{ANE2} and \cite{ANE1}. In \cite{ANE2} and \cite{ANE1}, AN is always modeled as residual interference at Eve. In contrast, this work considers an eavesdropper equipped with ANE capability. In the regime ${N_e} > {N_a} - {N_b}$, Eve can completely mitigate AN, and its achievable rate becomes independent of the AN power.}

\begin{remark}
From a mathematical perspective, Eq. (\ref{theo:3}) reveals that the normalized secrecy rate is the difference between two terms involving the $\Phi$ function. The first term, $\Phi({\alpha \gamma {N_b},{\delta _2}})$, represents Bob's normalized achievable rate, while the second term, $\Phi({\alpha {N_b},{\delta _1} - {\delta _2} + 1})$, corresponds to Eve's normalized achievable rate after applying the ANE scheme.

The function $\Phi(x,y)$ has the following key properties:
\begin{itemize}
\item It is monotonically increasing with respect to $x$ for fixed $y$, meaning higher SNR always improves achievable rates.
\item For fixed $x$, it increases with $y$ when $y > 1$, showing the benefit of having more transmit antennas.
\item For fixed $x$, as $y \to \infty$, $\Phi(x,y) \to \infty$, indicating unbounded capacity growth with increased transmit antennas.
\end{itemize}

\end{remark}

\begin{example}
Fig. \ref{fig_4} compares the analytical expression from (\ref{theo:3}) with simulation results. The parameters are deliberately set to $\alpha = 20$ dB, $\beta = 1$, $\gamma = 5$, $\delta_1 = 2.5$, and $\delta_2 = 2$ to satisfy $\delta_1 > \delta_2 - 1$. In this scenario, the formula predicts a positive but limited normalized secrecy rate. The simulation results show excellent agreement with the theoretical prediction as $N_b$ increases, validating that Eq. (\ref{theo:3}) accurately captures the secrecy performance in this regime.

\end{example}

\begin{figure}[t]
\centering
\includegraphics[width=3.5in,height=2.8in]{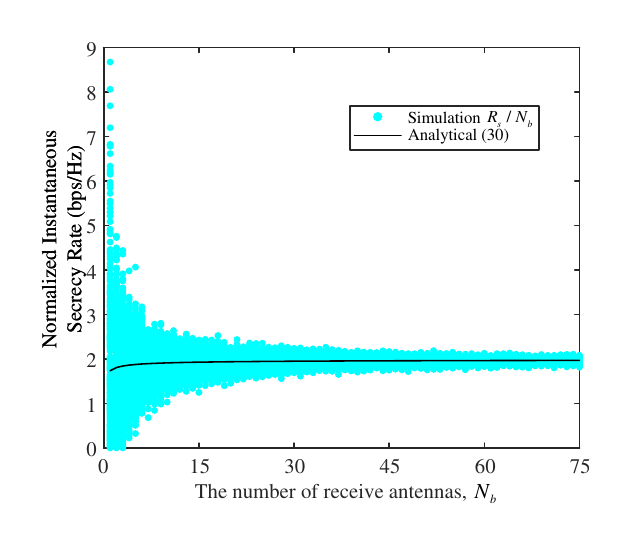}
\vspace{+3mm}
\caption{$R_s/N_b$ versus $N_b$ with $\alpha = 20$ dB, $\beta = 1$, $\gamma = 5$, $\delta_1 = 2.5$, and $\delta_2 = 2$.}
\label{fig_4}
\vspace{-0em}
\end{figure}

\begin{theorem}
As ${N_a},{N_b},{N_e} \to \infty $ with ${N_e}/{N_b} \to {\delta _1}$, ${N_a}/{N_b} \to {\delta _2}$, and ${\delta _1}  \leqslant  {\delta _2} - 1$, the normalized instantaneous secrecy rate approaches
\begin{equation}\label{theo:4}
\frac{{{R_s}}}{{{N_b}}} \to \Phi \left( {\alpha \gamma {N_b},{\delta _2}} \right).
\end{equation}
\end{theorem}

\begin{proof}
See Appendix D.
\end{proof}

\begin{remark}
Theorem 4 presents a fundamentally different mathematical result compared to Theorem 3. When ${\delta _1} \leqslant {\delta _2} - 1$, Eve's normalized achievable rate term vanishes completely, leaving only Bob's term in the expression. The mathematical explanation lies in that the single-dimensional ANE projection cannot provide an achievable rate comparable to the ${N_b}$ dimensional normalization used for Bob.

Specifically, when ${\delta _1} \leqslant {\delta _2} - 1$, Eve can identify a projection direction that effectively nullifies the power of the AN. However, since this direction offers only one DoF, normalization over ${N_b}$ dimensions causes her normalized achievable rate to approach zero.

This mathematical condition precisely characterizes the regime in which perfect secrecy (in the asymptotic sense) can be achieved, provided a sufficient number of antennas are available. 
\end{remark}

\begin{figure}[t]
\centering
\includegraphics[width=3.5in,height=2.8in]{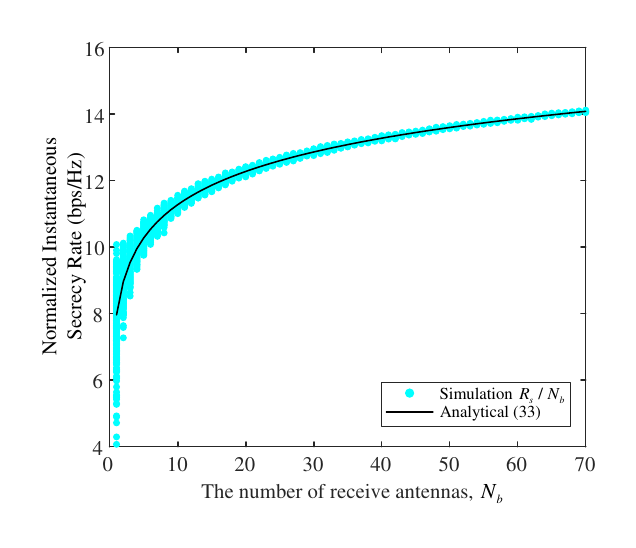}
\vspace{+3mm}
\caption{$R_s/N_b$ versus $N_b$ with $\alpha = 20$ dB, $\beta = 10$, $\gamma = 5$, $\delta_1 = 0.5$, and $\delta_2 = 1.5$.}
\label{fig_5}
\vspace{-0em}
\end{figure}

\begin{example}
	Fig. \ref{fig_5} compares the analytical results for ${R_s}/{N_b}$ from (\ref{theo:4}), with the corresponding simulation results, using the parameters $\alpha  = 20$ dB, $\beta  = 1$, $\gamma   = 1$, ${\delta _1}=1$, and ${\delta _2}=3$. As observed, as the number of receive antennas increases, the simulation results gradually converge to the analytical results.
\end{example}

When Eve adopts the ANE scheme, the combination of (\ref{theo:3}) and (\ref{theo:4}) redefines the scaling law of normalized instantaneous secrecy rate as
\begin{equation}
	\frac{{{R_s}}}{{{N_b}}} \to \left\{ {\begin{array}{*{20}{c}}
			{{{\left[ {\Phi \left( {\alpha \gamma {N_b},{\delta _2}} \right) - \Phi \left( {\alpha {N_b},{\delta _3}} \right)} \right]}^ + },{\text{ if }}{\delta _1} > {\delta _2} - 1,} \\
			{\Phi \left( {\alpha \gamma {N_b},{\delta _2}} \right), \ {\text{otherwise}}.} \ \ \ \ \ \ \ \ \ \ \ \ \ \ \ \ \ \ \ \ \ \ \ \
	\end{array}} \right.
\end{equation}

\section{Key Corollaries from Redefined Scaling Laws}

Based on the redefined scaling laws outlined above, we arrive at several straightforward corollaries. Specifically, we identify sufficient conditions for Eve to achieve perfect eavesdropping. Moreover, we approximate the average and instantaneous secrecy rates from the perspective of Alice. Furthermore, we provide a sufficient condition for achieving the secrecy capacity with Gaussian input alphabets.

\subsection{Sufficient Conditions for Perfect Eavesdropping}

\begin{tcolorbox}[colback=gray!5!white, colframe=black, title=Insights Into Eve's Antenna Threshold]

\begin{corollary}
	If Eve's AWGN power is comparable to or weaker than Bob's, i.e., $\gamma  \leqslant {N_b}/{N_a} $, and Eve is equipped with at least ${N_e} \geqslant 2{N_a} - {N_b}$ antennas, then the average secrecy rate vanishes, i.e., ${{\bar R}_s} = 0$.
\end{corollary}
\begin{proof}
	See Appendix E.
\end{proof}

\begin{remark}
	Corollary 1 demonstrates that when the power levels of AWGN at Bob and Eve are comparable, an Eve equipped with more than $2{N_a} - {N_b}$ antennas is capable of achieving perfect wiretapping, i.e., the ANE countermeasure prevents a positive average secrecy rate.
\end{remark}

\begin{example}
	Fig. \ref{fig_8} illustrates the average secrecy rate ${{\bar R}_s}$ as a function of $N_e$ with $\alpha  = 20$ dB, $\beta  = 1$, $\gamma   = 0.5$, and $N_b=8$. For the cases of ${N_a} = 14,15$, and 16, the minimum number of $N_e$ resulting in a zero ${{\bar R}_s}$ are 20, 22, and 24, respectively. This numerical relationship between $N_a$, $N_b$, and $N_e$ further validates the accuracy of Corollary 1.
\end{example}
\end{tcolorbox}

\begin{figure}[t]
	\centering
	\includegraphics[width=3.5in,height=2.8in]{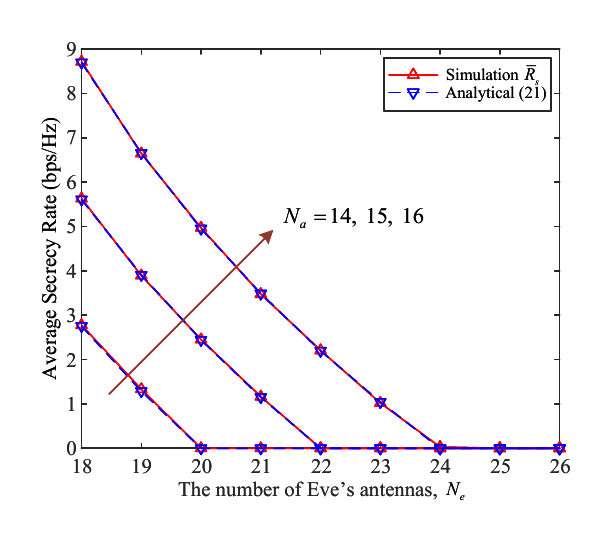}
	\vspace{+3mm}
	\caption{${{\bar R}_s}$ versus $N_e$ with $\alpha  = 20$ dB, $\beta  = 1$, $\gamma   = 0.5$, and $N_b=8$.}
	\label{fig_8}
	\vspace{-0em}
\end{figure}

\begin{corollary}
	Under the same assumptions as in Theorem 3, as ${\delta _1} = {\delta _2} \to 1$ with $\gamma  = 1$, we have ${R_s}/{N_b} \to 0$.
\end{corollary}
\begin{proof}
	Substituting ${\delta _1} = {\delta _2} = \gamma = 1$ into (\ref{theo:3}), the proof is completed.
\end{proof}

\begin{remark}
	Corollary 2 indicates that when the growth rates of $N_a$, $N_b$ and $N_e$ are identical, the normalized instantaneous secrecy rate ${R_s}/{N_b}$ tends to 0.
\end{remark}

\begin{example}
	Fig. \ref{fig_corollary2} plots the normalized instantaneous secrecy rate ${{R}_s}/{N_b}$ versus $N_b$ with $\alpha  = 20$ dB, $\beta  = 1$, $\gamma   = 1$, and ${N_a} = {N_e} = {N_b} + 1$. As expected, ${\delta _1}$ and $ {\delta _2}$ approach 1 with the increase of $N_b$, leading to ${{R}_s}/{N_b}$ asymptotically approaching 0.

The mathematical insight lies in the fact that when Eve has the same number of antennas as Alice, and her antenna count scales with the same growth trend as Bob's (but remains slightly larger), perfect eavesdropping becomes achievable in the asymptotic regime.
\end{example}

\begin{figure}[t]
	\centering
	\includegraphics[width=3.5in,height=2.8in]{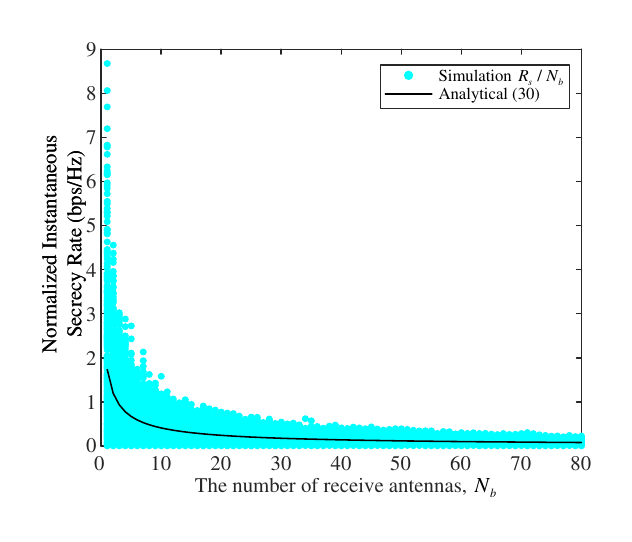}
	\vspace{+3mm}
	\caption{${{R}_s}/{N_b}$ versus $N_b$ with $\alpha  = 20$ dB, $\beta  = 1$, $\gamma   = 1$, and ${N_a} = {N_e} = {N_b} + 1$.}
	\label{fig_corollary2}
	\vspace{-0em}
\end{figure}

\begin{corollary}
	Under the same assumptions as in Theorem 3, as ${\delta _1} \to \infty $, we have ${R_s}/{N_b} \to 0$.
\end{corollary}
\begin{proof}
	See Appendix F.
\end{proof}

\begin{remark}
This corollary examines the extreme case where Eve has vastly more antennas than legitimate users. From (\ref{theo:3}), when $\delta_1 \to \infty$, the term $\Phi(\alpha N_b, \delta_1 - \delta_2 + 1)$ grows without bound due to the properties of   $\Phi$ function, eventually exceeding $\Phi(\alpha\gamma N_b, \delta_2)$, which causes $R_s/N_b \to 0$.

Mathematically, this represents the case where Eve's ANE advantage becomes overwhelming, allowing her to eliminate AN while preserving sufficient dimensions for signal recovery. The result confirms that no positive secrecy rate is achievable when Eve has unlimited antenna resources, regardless of other system parameters.
\end{remark}

\begin{figure}[t]
\centering
\includegraphics[width=3.5in,height=2.8in]{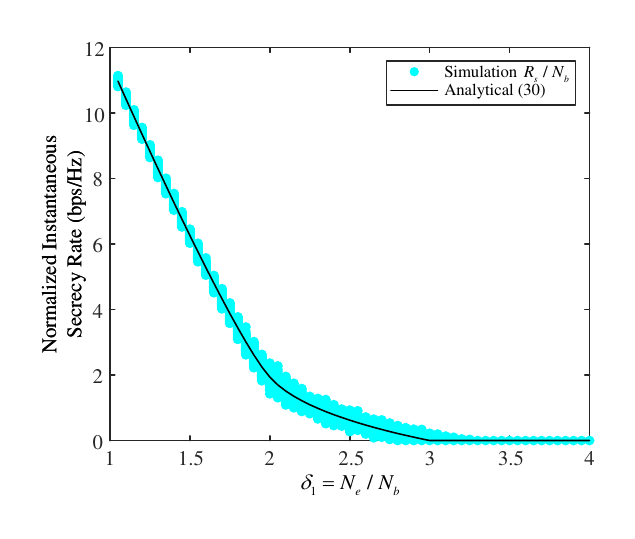}
\vspace{+3mm}
\caption{${{R}_s}/{N_b}$ versus ${\delta _1}$ with $\alpha  = 20$ dB, $\beta  = 1$, $\gamma   = 1$, ${N_a} = 40$, and ${N_b} = 20$.}
\label{fig_corollary3}
\vspace{-0em}
\end{figure}

\begin{example}
Fig. \ref{fig_corollary3} shows ${{R}_s}/{N_b}$ versus ${\delta _1}$ with $\alpha  = 20$ dB, $\beta  = 1$, $\gamma   = 1$, ${N_a} = 40$, and ${N_b} = 20$. As anticipated, ${{R}_s}/{N_b}$ asymptotically tends to 0 as ${\delta _1}$ increases, validating the accuracy of Corollary 3.
\end{example}

\begin{corollary}
Under the same assumptions as in Theorem 3, if ${\delta _1} = 2{\delta _2}-1$ and $\gamma = 1$, then ${R_s}/{N_b} \to 0$.
\end{corollary}
\begin{proof}
Substituting ${\delta _1} = 2{\delta _2}-1$ and $\gamma = 1$ into (\ref{theo:3}), the proof is completed.
\end{proof}

\begin{remark}
Corollary 4 reveals that when the growth rate of $N_e$ is twice that of $N_a$, the normalized instantaneous secrecy rate ${R_s}/{N_b}$ tends to 0. Additionally, Corollary 4 can be viewed as an extension of Corollary 1 in the case where the number of antennas approaches infinity.

The mathematical insight lies in identifying $\delta_1 = 2 \delta_2 - 1$ as a critical threshold, derived from equating Eve's term $\Phi({\alpha {N_b},{\delta _1} - {\delta _2} + 1})$ with Bob's term $\Phi({\alpha \gamma {N_b},{\delta _2}})$. This condition precisely marks the boundary beyond which, when $\delta_1  \geqslant  2 \delta_2 - 1$, Eve can achieve perfect eavesdropping in the asymptotic sense.
\end{remark}

\begin{example}
Fig. \ref{fig_corollary4} depicts ${{R}_s}/{N_b}$ as a function of ${\delta _1}$ and ${\delta _2}$, with the parameters set to $\alpha  = 10$ dB, $\beta  = 10$, $\gamma   = 1$, and ${N_b} = 50$. The $x$ and $y$ axes represents the values of ${\delta _1}$ and ${\delta _2}$, respectively, while the gradient corresponds to the values of ${R_s}/{N_b}$. As observed, in the region to the upper left of the line ${\delta _1} = 2{\delta _2} - 1$, ${R_s}/{N_b}$ approaches 0, indicating that the system cannot ensure secrecy performance.
\end{example}

\begin{figure}[t]
\centering
\includegraphics[width=3.5in,height=2.8in]{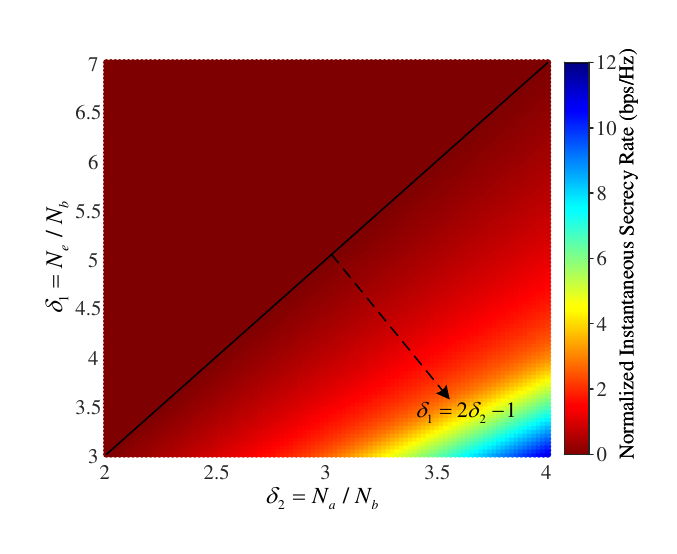}
\vspace{+3mm}
\caption{${{R}_s}/{N_b}$ versus ${\delta _1}$ and ${\delta _2}$ with $\alpha  = 10$ dB, $\beta  = 10$, $\gamma   = 1$, and ${N_b} = 50$.}
\label{fig_corollary4}
\vspace{-0em}
\end{figure}

%
%
%

\subsection{Approximation of Secrecy Rates}

For practical antenna configurations where $N_a$, $N_b$, and $N_e$ are finite, closed-form approximations of secrecy rates are useful for system design. Here, we derive mathematically tractable expressions that capture the performance with high accuracy while providing insight into the underlying mechanisms.

\begin{corollary}\label{co:1}
The average and instantaneous secrecy rates can be approximated as
\begin{equation}\label{approximation:1}
\begin{gathered}
  {{\bar R}_s},{R_s} \approx R_s^{app} = \hfill \\
   \left\{ {\begin{array}{*{20}{c}}
  {{N_b}\Phi \left( {\alpha \gamma {N_b},\frac{{{N_a}}}{{{N_b}}}} \right) - \log \left( {1 + \frac{{\alpha {N_b}}}{{\alpha \beta \mu  + 1}}} \right){\text{, if }}{N_e} \leqslant {N_a} - {N_b}} \\
  {{N_b}{{\left[ {\Phi \left( {\alpha \gamma {N_b},\frac{{{N_a}}}{{{N_b}}}} \right) - \Phi \left( {\alpha {N_b},\frac{{{N_e} - {N_a} + {N_b}}}{{{N_b}}}} \right)} \right]}^ + },{\text{otherwise}},}
\end{array}} \right. \hfill \\
\end{gathered}
\end{equation}
where $\mu  = \int_0^\infty  {B{f_B}\left( b \right){\text{d}}b}$ denotes the mean of the minimum eigenvalue
of the matrix ${\mathbf{G}}{{\mathbf{V}}_0}{\mathbf{V}}_0^H{{\mathbf{G}}^H}$.
\end{corollary}

\begin{proof}
See Appendix G.
\end{proof}

{
Although (\ref{approximation:1}) is presented for both the average and instantaneous secrecy rates, it should be interpreted as an approximation with different meanings. The instantaneous secrecy rate is a random variable that varies across channel realizations, and thus any closed-form expression can only approximate its typical behavior rather than provide an exact value for each realization.

The average secrecy rate corresponds to the expectation of the instantaneous secrecy rate and therefore reflects its central tendency. Since the approximation in (\ref{approximation:1}) characterizes the dominant deterministic component of the instantaneous secrecy rate, it naturally lies within the variation range of the instantaneous secrecy rate and closely approximates its average value.

Therefore, the same expression can be employed as a practical approximation for both the instantaneous and average secrecy rates, which explains the applicability of Corollary \ref{co:1} to both metrics.

Note that Corollary \ref{co:1} is presented as an approximation, as it extends an asymptotic characterization to the finite-antenna regime. Although deriving an explicit error bound is analytically challenging and thus left for future work, the subsequent numerical results demonstrate that the proposed approximation provides an accurate characterization of the average behavior.}

\begin{remark}
The approximation in (\ref{approximation:1}) bridges the finite-antenna case with the asymptotic results in Theorems 3 and 4. Mathematically, this approximation has the following properties:

\begin{itemize}
\item When ${N_e} > {N_a} - {N_b}$, the expression directly adopts the asymptotic form from Theorem 3, applying the finite-antenna values to the $\Phi $ function.
\item In the regime of ${N_e} \leqslant {N_a} - {N_b}$, the achievable rate at Eve is represented by a modified logarithmic term $\log \left( {1 + \alpha {N_b}/\left( {\alpha \beta \mu  + 1} \right)} \right)$ instead of a $\Phi$ function.
\item The parameter $\mu$ captures the average minimum eigenvalue effect, accounting for the random channel variations that are averaged out in the asymptotic case.
\end{itemize}

This approximation successfully transforms the complex statistical expectation in (\ref{isz:3}) into a deterministic function, making performance prediction and system optimization tractable.
\end{remark}

\begin{example}
Figs. \ref{fig_6} and \ref{fig_7} compare the analytical approximations from Eq. (\ref{approximation:1}) with simulation results in two different antenna regimes. Fig. \ref{fig_6} addresses the ${N_e} > {N_a} - {N_b}$ case with $\alpha = 20$ dB, $\beta = 1$, $N_a=32$, $N_b=16$, and $N_e=40$. In this configuration, the antenna condition ${N_e} > {N_a} - {N_b}$ is satisfied, placing the system in the regime where Eve can effectively apply ANE.

The simulation results confirm the theoretical prediction across the entire range of $\gamma$ values, with both the average and instantaneous secrecy rates closely matching the analytical curve. This validation demonstrates that the approximation in (\ref{approximation:1}) accurately captures the system behavior in this challenging scenario where Eve possesses significant antenna advantage.
\end{example}

\begin{figure}[t]
\centering
\includegraphics[width=3.5in,height=2.8in]{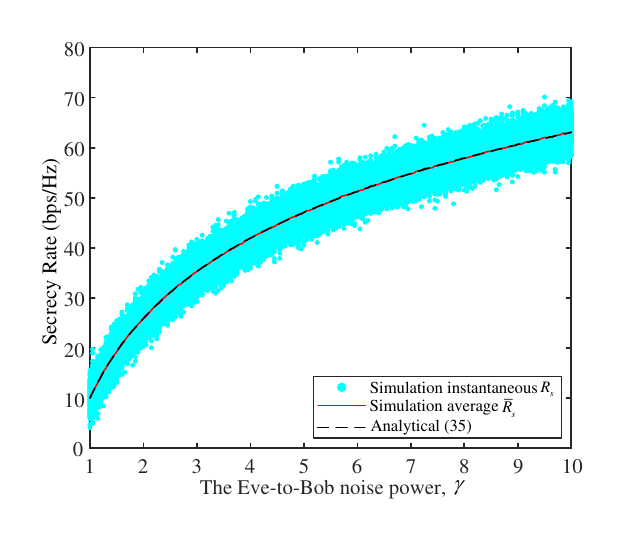}
\vspace{+3mm}
\caption{${{\bar R}_s},{R_s}$ versus $\gamma$ with $\alpha = 20$ dB, $\beta = 1$, $N_a=32$, $N_b=16$, and $N_e=40$ (${N_e} > {N_a} - {N_b}$ regime).}
\label{fig_6}
\vspace{-0em}
\end{figure}

\begin{remark}
The structure of (\ref{approximation:1}) reveals a key insight: when ${N_e} \leqslant {N_a} - {N_b}$, the parameter $\beta$ (AN power ratio) appears in the denominator of Eve's term as $\alpha \beta \mu + 1$. This confirms mathematically that increasing AN power always improves secrecy performance in this regime. Conversely, when ${N_e} > {N_a} - {N_b}$, $\beta$ is absent in Eve's term, indicating that the ANE scheme perfectly mitigates the influence of AN.
\end{remark}

\begin{example}
Fig. \ref{fig_7} examines the ${N_e} \leqslant {N_a} - {N_b}$ case with $\beta = 1$, $\gamma = 1$, $N_a=32$, $N_b=16$, and $N_e=12$. As shown in Fig. \ref{fig_7}, both average and instantaneous secrecy rates increase monotonically with $\alpha$, confirming that higher transmit power consistently improves secrecy performance in this antenna regime. The analytical approximation maintains high accuracy across the entire SNR range, demonstrating its effectiveness for practical system design.
\end{example}

\begin{figure}[t]
\centering
\includegraphics[width=3.5in,height=2.8in]{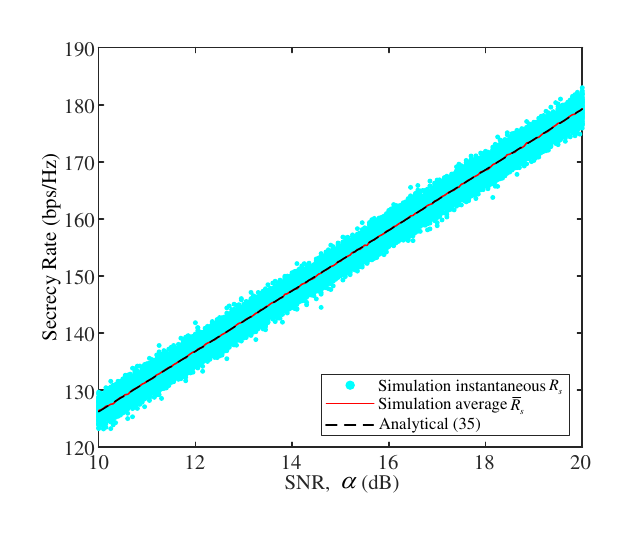}
\vspace{+3mm}
\caption{${{\bar R}_s},{R_s}$ versus $\alpha$ with $\beta = 1$, $\gamma = 1$, $N_a=32$, $N_b=16$, and $N_e=12$ (${N_e} \leqslant {N_a} - {N_b}$ regime).}
\label{fig_7}
\vspace{-0em}
\end{figure}

\subsection{Achieving the Average Secrecy Capacity}

The following corollary presents a sufficient condition for achieving the average secrecy capacity using Gaussian input alphabets, providing a theoretical upper bound on the secrecy performance.

\begin{corollary}
If $\beta \to \infty$ and ${N_e} \leqslant {N_a} - {N_b}$ are satisfied, then the inequality ${\bar R_s} \leqslant {\bar C_s} \leqslant {\bar C_b}$ asymptotically becomes an equality, where ${\bar C_b}$ denotes the average channel capacity for Bob.
\end{corollary}

\begin{proof}
See Appendix H.
\end{proof}

\begin{remark}
From a mathematical perspective, this corollary establishes that under specific conditions, the achievable secrecy rate approaches the fundamental limit. The condition $\beta \to \infty$ implies unlimited AN power, which drives Eve's achievable rate to zero according to
\begin{equation}\label{secrecy_cap:1}
\begin{gathered}
  \ \ \ \mathop {\lim }\limits_{\beta  \to \infty } {{\text{E}}_{{\mathbf{H}},{\mathbf{G}}}}\left[ {I\left( {{\mathbf{s}};{\mathbf{z}}|{\mathbf{H}},{\mathbf{G}}} \right)} \right] \hfill \\
   = \mathop {\lim }\limits_{\beta  \to \infty } {{\text{E}}_{\mathbf{g}}}\left( {\log \left( {1 + \frac{{\alpha {{\mathbf{g}}^H}{\mathbf{g}}}}{{\alpha \beta {\lambda _{{{\min }^2}}} + 1}}} \right)} \right) = 0, \hfill \\
\end{gathered}
\end{equation}
where $\lambda_{\min}$ is the minimum eigenvalue of ${\mathbf{V}}_0^H{{\mathbf{G}}^H}$. When ${N_e} \leqslant {N_a} - {N_b}$, this matrix has full rank with $\lambda_{\min} \neq 0$, ensuring that $\alpha {{\mathbf{g}}^H}{\mathbf{g}}/\left( {\alpha \beta \lambda _{\min }^2 + 1} \right) \to 0$ as $\beta \to \infty$.

Consequently, the secrecy rate simplifies to
\begin{equation}\label{secrecy_cap:2}
\begin{gathered}
  {{\bar R}_s} = {{\text{E}}_{\mathbf{H}}}\left[ {I\left( {{\mathbf{s}};{\mathbf{y}}|{\mathbf{H}}} \right)} \right] - {{\text{E}}_{{\mathbf{H}},{\mathbf{G}}}}\left[ {I\left( {{\mathbf{s}};{\mathbf{z}}|{\mathbf{H}},{\mathbf{G}}} \right)} \right] \hfill \\
  \ \ \ \ = {{\text{E}}_{\mathbf{H}}}\left[ {I\left( {{\mathbf{s}};{\mathbf{y}}|{\mathbf{H}}} \right)} \right] = {{\bar C}_b}, \hfill \\
\end{gathered}
\end{equation}
indicating that Alice can transmit to Bob at the full channel capacity with perfect secrecy. This represents the optimal operating point from an information-theoretic perspective, where Gaussian input alphabets achieve the secrecy capacity.
\end{remark}

\begin{figure}[t]
\centering
\includegraphics[width=3.5in,height=2.8in]{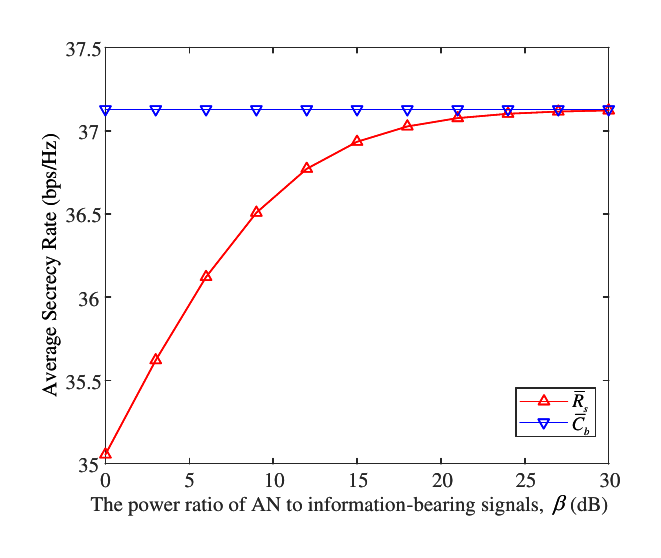}
\vspace{+3mm}
\caption{${{\bar R}_s}$ and ${{\bar C}_b}$ versus $\beta$ (dB) with $\alpha = 3$ dB, $\gamma = 1$, ${N_a} = 16$, ${N_b} = 8$, and ${N_e} = 4$ (${N_e} \leqslant {N_a} - {N_b}$ regime).}
\label{fig_corollary6}
\vspace{-0em}
\end{figure}

\begin{example}
Fig. \ref{fig_corollary6} demonstrates the convergence of ${{\bar R}_s}$ to ${{\bar C}_b}$ as $\beta$ increases. The parameters are set to $\alpha = 3$ dB, $\gamma = 1$, ${N_a} = 16$, ${N_b} = 8$, and ${N_e} = 4$, satisfying the crucial antenna condition ${N_e} < {N_a} - {N_b}$.

The mathematical prediction is that as $\beta$ increases, the secrecy rate converges to Bob's channel capacity. Based on equation (\ref{secrecy_cap:1}), the gap between ${{\bar R}_s}$ and ${{\bar C}_b}$ diminishes progressively as $\beta$ increases, with the rate of convergence determined by
\begin{align}\label{secrecy_cap:3}
{\bar C_b} - {\bar R_s} &= {{\rm E}_{\mathbf{g},\lambda_{\min}}}\left({\log \left( {1 + \frac{{\alpha {{\mathbf{g}}^H}{\mathbf{g}}}}{{\alpha \beta \lambda_{\min}^2 + 1}}} \right)}\right) \nonumber\\
&\to 0, \text{ as } \beta \to \infty.
\end{align}

This result demonstrates a fundamental information-theoretic principle: with sufficient AN power in the regime of ${N_e} \leqslant {N_a} - {N_b}$, near-optimal secrecy performance can be achieved. The simulation validates that the theoretical conditions in Corollary 6 are sufficient for approaching the secrecy capacity bound.
\end{example}

%

\section{Performance Without the AN Scheme}

The preceding discussions have demonstrated that while the AN scheme offers significant security benefits, the ANE scheme can mitigate the impact of AN under certain conditions, particularly when Eve's hardware capabilities significantly exceed those of legitimate parties. Having characterized the performance boundaries when AN is used, a natural question arises: \textit{Does the emergence of the ANE scheme render AN ineffective in enhancing PLS?} To address this fundamental question, we now provide a mathematical analysis of secrecy rates in the absence of the AN scheme, establishing a benchmark for comparison with our earlier results.

\subsection{Average Secrecy Rate Analysis}

Building on the same system model described in Section II, but removing the AN component, we establish the following result:

\begin{theorem}
In the absence of the AN scheme, the average secrecy rate can be formulated as
\begin{equation}\label{theo:5}
\bar R_s^a = {\left[ {f\left( {{N_a},{N_b},\eta } \right) - f\left( {{N_a},{N_e},\eta /\gamma } \right)} \right]^ + },
\end{equation}
where $\eta = \alpha \gamma {N_a}\left( {1 + \beta \left( {{N_a} - {N_b}} \right)/{N_b}} \right)$ denotes the equivalent SNR, and $f(t,r,x)$ is defined in (\ref{func:1}).
\end{theorem}

\begin{proof}
See Appendix I.
\end{proof}

\begin{remark}
The mathematical structure of (\ref{theo:5}) reveals several key insights:

\begin{itemize}
\item The secrecy rate is expressed as the difference between Bob's and Eve's achievable rates, with both terms using the same function $f(\cdot)$ but with different parameters.
\item The equivalent SNR $\eta$ incorporates the total power that would otherwise be split between information signal and AN.
\item The term $\eta/\gamma$ in Eve's achievable rate reflects her channel quality relative to Bob's.
\item When $f({{N_a},{N_e},\eta /\gamma}) \geq f({{N_a},{N_b},\eta})$, the secrecy rate becomes zero, indicating perfect eavesdropping.
\end{itemize}

For a fair comparison with the AN scheme, Eq. (\ref{theo:5}) assumes that the total power, initially allocated for transmitting both the information-bearing signal and AN, is now solely used for transmitting the information-bearing signal.
\end{remark}

\begin{figure}[t]
\centering
\includegraphics[width=3.5in,height=2.8in]{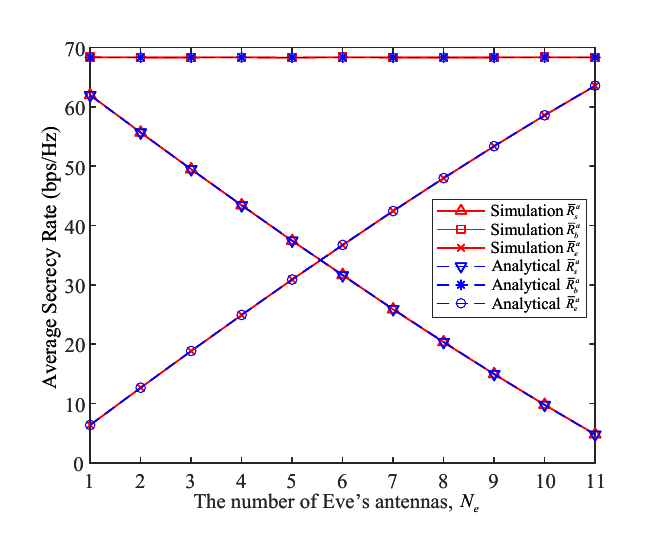}
\vspace{+3mm}
\caption{$\bar R_s^a$, $\bar R_b^a$, and $\bar R_e^a$ versus $N_e$ with $\alpha = 6$ dB, $\beta = 1$, $\gamma = 1$, ${N_a} = 16$, and ${N_b} = 12$.}
\label{fig_theorem5}
\vspace{-0em}
\end{figure}

\begin{example}
Fig. \ref{fig_theorem5} plots the analytical expressions for $\bar R_s^a$, $\bar R_b^a$, and $\bar R_e^a$ from (\ref{theo:5}) alongside the corresponding simulation results. The parameters are set to $\alpha = 6$ dB, $\beta = 1$, $\gamma = 1$, ${N_a} = 16$, and ${N_b} = 12$.

The mathematical behavior observed in Fig. \ref{fig_theorem5} confirms that as $N_e$ increases, $\bar R_e^a$ grows monotonically while $\bar R_b^a$ remains constant, resulting in a diminishing $\bar R_s^a$. As expected, the analytical average secrecy rate and achievable rates, i.e., $\bar R_s^a$, $\bar R_b^a$, and $\bar R_e^a$, are in close agreement with the simulated values.
\end{example}

\subsection{Asymptotic Behavior Analysis}

Similar to our approach in Section III. B., we now examine the asymptotic behavior of the secrecy rate without AN to enable direct comparison with the AN-based results in Theorems 3 and 4:

\begin{theorem}
As ${N_a},{N_b},{N_e} \to \infty $ with ${N_e}/{N_b} \to {\delta _1}$ and ${N_a}/{N_b} \to {\delta _2}$, the normalized instantaneous secrecy rate without AN approaches
\begin{equation}\label{theo:6}
\frac{{R_s^a}}{{{N_b}}} \to {\left[ {\Phi \left( {\frac{\eta }{{{\delta _2}}},{\delta _2}} \right) - {\delta _1}\Phi \left( {\frac{{\eta {\delta _1}}}{{\gamma {\delta _2}}},\frac{{{\delta _2}}}{{{\delta _1}}}} \right)} \right]^ + },
\end{equation}
where $\Phi \left( {x,y} \right)$ is defined in (\ref{fidef:1}) and $\eta$ represents the equivalent SNR.
\end{theorem}

\begin{proof}
See Appendix J.
\end{proof}

\begin{remark}
Eq. (\ref{theo:6}) provides critical insight into the asymptotic security performance without AN. Comparing this result with Theorems 3 and 4, we observe a fundamental difference: the mathematical structure differs significantly from the AN case, particularly in how Eve's channel contribution is represented. Unlike the AN case where Eve's term vanishes when $\delta_1 \leq \delta_2 - 1$, here the second term ${\delta _1}\Phi (\eta {\delta _1}/\left( {\gamma {\delta _2}} \right),{\delta _2}/{\delta _1})$ is always positive for any non-zero $\delta_1$. This difference highlights a key security advantage of the AN scheme.

\begin{itemize}
\item The first term $\Phi (\eta /{\delta _2},{\delta _2})$ represents Bob's normalized achievable rate, similar to the AN case.
\item The second term ${\delta _1}\Phi (\eta {\delta _1}/\left( {\gamma {\delta _2}} \right),{\delta _2}/{\delta _1})$ represents Eve's normalized achievable rate, scaled by her antenna ratio $\delta_1$.
\item Unlike the AN case with ANE, there is no critical threshold that eliminates Eve's term regardless of $\delta_1$.
\end{itemize}

This theorem explicitly establishes that without AN, Eve can always achieve a positive normalized achievable rate when $\delta_1 > 0$, contrasting sharply with the AN case where Eve's achievable rate vanishes when $\delta_1 \leq \delta_2 - 1$ as shown in Theorem 4.
\end{remark}

\begin{figure}[t]
\centering
\includegraphics[width=3.5in,height=2.8in]{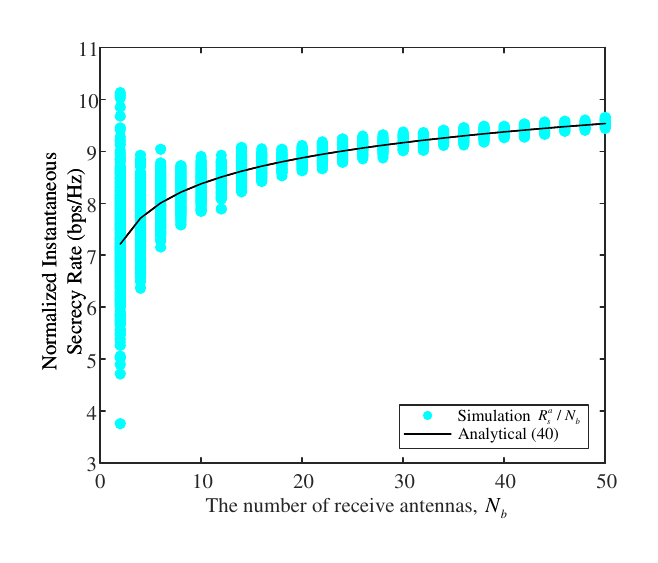}
\vspace{+3mm}
\caption{$R_s^a/{N_b}$ versus $N_b$ with $\alpha = 20$ dB, $\beta = 10$, $\gamma = 5$, ${\delta_1} = 0.5$, and ${\delta_2} = 1.5$.}
\label{fig_theorem6}
\vspace{-0em}
\end{figure}

\begin{example}
Fig. \ref{fig_theorem6} compares the analytical expression for $R_s^a/{N_b}$ from (\ref{theo:6}) with simulation results. The parameters are selected as $\alpha = 20$ dB, $\beta = 10$, $\gamma = 5$, ${\delta_1} = 0.5$, and ${\delta_2} = 1.5$.
The simulation confirms the theoretical prediction, with convergence becoming more precise as $N_b$ increases.
\end{example}

\subsection{Conditions for Perfect Eavesdropping}

To gain deeper insight into when security completely fails without AN, we identify mathematical conditions that result in zero secrecy rate and compare them with the conditions for failure derived in Corollary 1 for the AN case:

\begin{corollary}
If $\gamma \leqslant 1$ and $N_e \geqslant N_b$ are satisfied, then $\bar R_s^a = 0$.
\end{corollary}

\begin{proof}
Substituting $\gamma = 1$ and $N_e = N_b$ into (\ref{theo:5}), and applying the properties (\ref{pp:2}) and (\ref{pp:3}), we obtain
\begin{align}
f(N_a,N_e,\eta/\gamma) &= f(N_a,N_b,\eta), \ \text{for } \gamma = 1, N_e = N_b, \nonumber\\
f(N_a,N_e,\eta/\gamma) &\geq f(N_a,N_b,\eta), \ \text{for } \gamma \leq 1, N_e \geq N_b,
\end{align}
resulting in $\bar R_s^a = 0$.
\end{proof}

\begin{remark}
This corollary provides a simple but powerful mathematical characterization: when the power levels of AWGN at Bob and Eve are similar ($\gamma \leq 1$), Eve requires only $N_e \geq N_b$ antennas to achieve perfect eavesdropping.

Comparing this with Corollary 1 for the AN case, where Eve requires $N_e \geq 2N_a - N_b$ antennas, we observe a significant security advantage of the AN scheme. The introduction of AN increases the antenna threshold from $N_b$ to $2N_a - N_b$, creating an antenna gap of $2{N_a} - 2{N_b}$ that Eve must overcome (note that ${N_a} > {N_b}$ naturally arises from the requirement for generating AN). This mathematical comparison quantitatively demonstrates the effectiveness of AN in enhancing security, directly addressing the question posed at the beginning of this section.
\end{remark}

\begin{figure}[t]
\centering
\includegraphics[width=3.5in,height=2.8in]{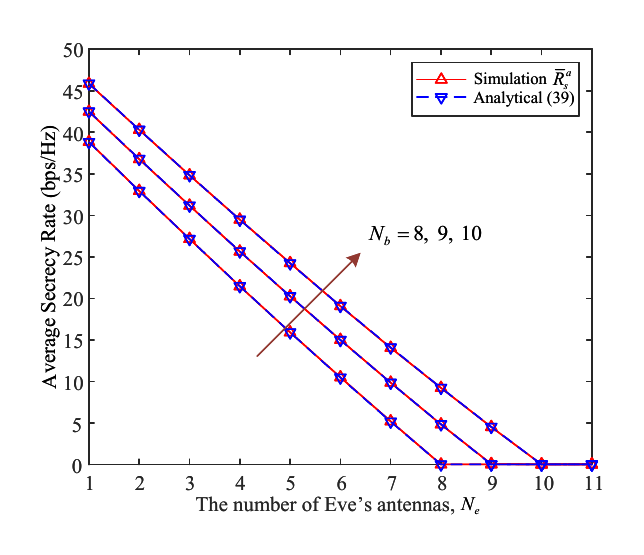}
\vspace{+3mm}
\caption{$\bar R_s^a$ versus $N_e$ with $\alpha = 3$ dB, $\beta = 1$, $\gamma = 1$, and ${N_a} = 16$ for different values of $N_b$.}
\label{fig_corollary7}
\vspace{-0em}
\end{figure}

\begin{example}
Fig. \ref{fig_corollary7} illustrates the average secrecy rate $\bar R_s^a$ as a function of $N_e$ with parameters $\alpha = 3$ dB, $\beta = 1$, $\gamma = 1$, and ${N_a} = 16$. Three curves are shown for ${N_b} = 8$, $9,$ and $10$.

The results confirm that for each case, $\bar R_s^a$ becomes exactly zero when $N_e$ reaches $N_b$. Specifically, the critical points occur at $N_e = 8$, $9,$ and $10$ respectively, validating the mathematical prediction of Corollary 7. The abrupt transition to zero secrecy rate demonstrates that without AN, security is completely compromised once Eve's antenna count reaches the threshold $N_b$.
\end{example}

\begin{corollary}
Under the same assumption of Theorem 6, as ${\delta _1} \to 1$ with $\gamma = 1$, we have $R_s^a/{N_b} \to 0$.
\end{corollary}

\begin{proof}
Substituting ${\delta _1} = 1$ and $\gamma = 1$ into (\ref{theo:6}), the proof is completed.
\end{proof}

\begin{remark}
Corollary 8 extends the finite antenna result to the asymptotic case: when the growth rate of $N_e$ approaches that of $N_b$ with equal noise power levels, the normalized instantaneous secrecy rate vanishes. This establishes a fundamental limit on secrecy performance in massive MIMO systems without AN, regardless of how large the antenna arrays become.
\end{remark}

\begin{figure}[t]
\centering
\includegraphics[width=3.5in,height=2.8in]{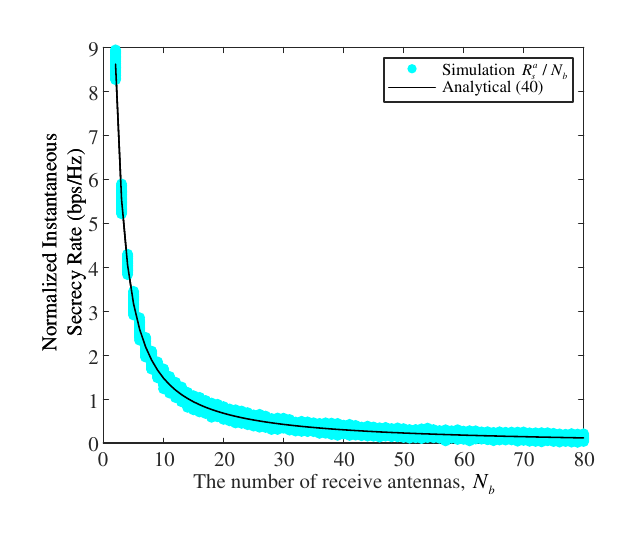}
\vspace{+3mm}
\caption{$R_s^a/{N_b}$ versus $N_b$ with $\alpha = 15$ dB, $\beta = 1$, $\gamma = 1$, ${N_a} = 100$, and ${N_e} = {N_b}-1$.}
\label{fig_corollary8}
\vspace{-0em}
\end{figure}

\begin{example}
Fig. \ref{fig_corollary8} depicts $R_s^a/{N_b}$ versus $N_b$ with parameters $\alpha = 15$ dB, $\beta = 1$, $\gamma = 1$, ${N_a} = 100$, and ${N_e} = {N_b}-1$. As $N_b$ increases, $\delta_1 = N_e/N_b = (N_b-1)/N_b$ converges to 1, causing the normalized instantaneous secrecy rate to approach zero. This confirms the mathematical prediction of Corollary 8 and illustrates how even a slight numerical advantage in antenna count becomes irrelevant in the asymptotic regime.
\end{example}

\subsection{Approximation of Secrecy Rates}

Following our approach in Section IV. B. for the AN case, we now develop a comparable approximation for the no-AN scenario to facilitate direct performance comparisons:

\begin{corollary}
The average and instantaneous secrecy rates without the AN scheme can be approximated as
\begin{equation}\label{corollary:7}
R_{s,app}^a = {\left[ {{N_b}\Phi \left( {\frac{{\eta {N_b}}}{{{N_a}}},\frac{{{N_a}}}{{{N_b}}}} \right) - {N_e}\Phi \left( {\frac{{\eta {N_e}}}{{\gamma {N_a}}},\frac{{{N_a}}}{{{N_e}}}} \right)} \right]^ + }.
\end{equation}
\end{corollary}

\begin{proof}
The proof follows directly from (\ref{theo:6}) by applying finite antenna values.
\end{proof}

\begin{remark}
Eq. (\ref{corollary:7}) transforms the complex function $f(\cdot)$ into the more tractable $\Phi(\cdot)$ function, offering significant computational advantages. The structure of this approximation reveals that:

\begin{itemize}
\item Bob's achievable rate scales with $N_b$, while Eve's scales with $N_e$.
\item The ratio ${N_a}/{N_b}$ determines Bob's spatial multiplexing gain.
\item The ratio ${N_a}/{N_e}$ determines Eve's spatial multiplexing gain.
\end{itemize}
\end{remark}

\begin{figure}[t]
\centering
\includegraphics[width=3.5in,height=2.8in]{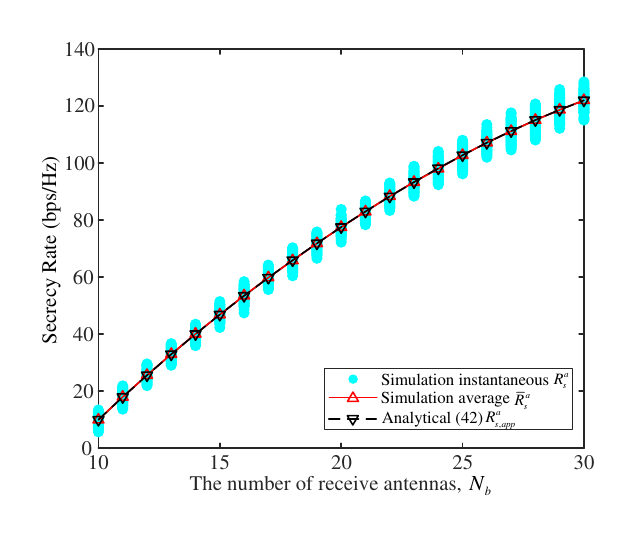}
\vspace{+3mm}
\caption{$\bar R_s^a$ and $R_s^a$ versus $N_b$ with $\alpha = 3$ dB, $\beta = 1$, $\gamma = 2$, ${N_a} = 32$, and ${N_e} = 10$.}
\label{fig_corollary9}
\vspace{-0em}
\end{figure}

\begin{example}
Fig. \ref{fig_corollary9} plots the analytical approximations $R_{s,app}^a$ from (\ref{corollary:7}) alongside the simulated secrecy rates $\bar R_s^a$ and $R_s^a$. The parameters are set to $\alpha = 3$ dB, $\beta = 1$, $\gamma = 2$, ${N_a} = 32$, and ${N_e} = 10$.

The close match between the approximation and simulation across the entire range of $N_b$ values confirms the accuracy of (\ref{corollary:7}). Notably, both the average and instantaneous secrecy rates increase with $N_b$, as a larger receive antenna array at Bob strengthens his advantage over Eve. The approximation successfully captures this trend, providing a computationally efficient tool for system designers.
\end{example}

\subsection{Comparative Analysis: With vs. Without AN}

To quantify the security benefits of AN, we identify precise conditions under which the AN scheme outperforms the no-AN approach:

\begin{corollary}
${N_b} \leqslant {N_e} < 2{N_a} - {N_b}$ combined with ${N_b}/{N_a} < \gamma \leqslant 1$ is a sufficient condition for $ 0 = \bar R_s^a < {\bar R_s}$.
\end{corollary}

\begin{proof}
See Appendix K.
\end{proof}

\begin{remark}
This corollary identifies a critical region in the parameter space where AN provides a decisive security advantage. When ${N_b} \leqslant {N_e} < 2{N_a} - {N_b}$ and ${N_b}/{N_a} < \gamma \leqslant 1$, the no-AN scheme results in zero secrecy rate (perfect eavesdropping), while employing AN ensures a positive secrecy rate despite Eve's ANE countermeasures.

This region constitutes the ``AN advantage zone" where even imperfect AN implementation outperforms the complete absence of AN. The size of this zone grows with the ratio $N_a/N_b$, highlighting that the security benefit of AN increases with the transmit-to-receive antenna ratio at the legitimate parties.
\end{remark}

\begin{figure}[t]
\centering
\includegraphics[width=3.5in,height=2.8in]{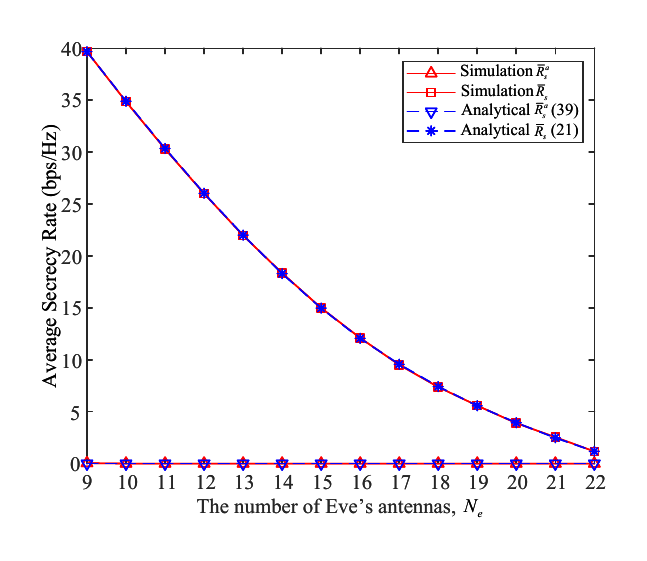}
\vspace{+3mm}
\caption{$\bar R_s^a$ and $\bar R_s$ versus $N_e$ with $\alpha = 6$ dB, $\beta = 1$, $\gamma = 1$, ${N_a} = 16$, and ${N_b} = 9$.}
\label{fig_corollary10}
\vspace{-0em}
\end{figure}

\begin{example}
Fig. \ref{fig_corollary10} provides a direct comparison between $\bar R_s^a$ (without AN) and $\bar R_s$ (with AN) as functions of $N_e$. The parameters are set to $\alpha = 6$ dB, $\beta = 1$, $\gamma = 1$, ${N_a} = 16$, and ${N_b} = 9$, with $N_e$ varying from 9 to 22, ensuring ${N_b} \leqslant {N_e} < 2{N_a} - {N_b}$.

As shown in Fig. \ref{fig_corollary10}, $\bar R_s^a$ remains zero across all cases, representing complete security failure without AN. In contrast, $\bar R_s$ maintains a positive value that gradually decreases as $N_e$ increases. This dramatic difference confirms the mathematical prediction of Corollary 10 and demonstrates the profound security advantage of AN in this parameter region.
\end{example}

\begin{corollary}
${N_e} < {N_b}$ combined with $\beta \to 0$ and $\gamma > 1$ is a sufficient condition for $0 < \bar R_s^a < {\bar R_s}$.
\end{corollary}

\begin{proof}
See Appendix L.
\end{proof}

\begin{remark}
This corollary complements Corollary 10 by addressing the case where Eve has fewer antennas than Bob. It establishes that a marginal amount of AN power ($\beta \to 0$) enhances security when Bob's channel quality exceeds that of Eve ($\gamma > 1$). This is particularly noteworthy as it highlights that AN contributes to PLS across a larger parameter space, not just in the critical region identified by Corollary 10.
\end{remark}

\begin{figure}[t]
\centering
\includegraphics[width=3.5in,height=2.8in]{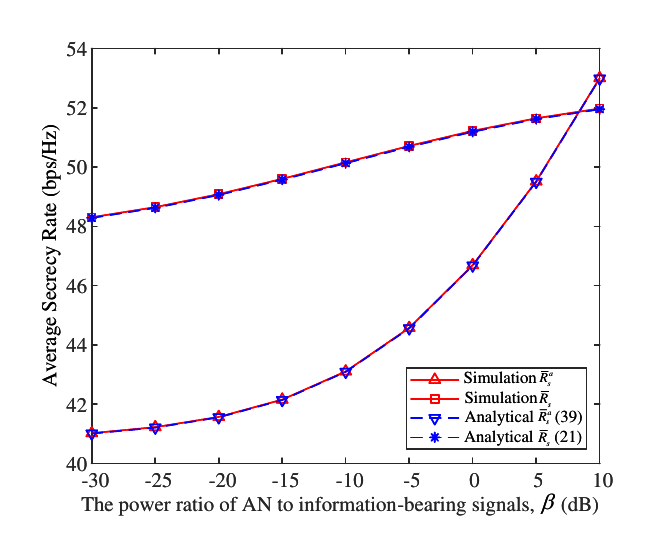}
\vspace{+3mm}
\caption{$\bar R_s^a$ and $\bar R_s$ versus $\beta$ (dB) with $\alpha = 6$ dB, $\gamma = 2$, ${N_a} = 16$, ${N_b} = 8$, and ${N_e} = 2$.}
\label{fig_corollary11}
\vspace{-0em}
\end{figure}

\begin{example}
Fig. \ref{fig_corollary11} presents $\bar R_s^a$ and $\bar R_s$ as functions of $\beta$ (dB) with parameters $\alpha = 6$ dB, $\gamma = 2$, ${N_a} = 16$, ${N_b} = 8$, and ${N_e} = 2$.

Both rates increase with $\beta$, but for different reasons: $\bar R_s^a$ increases due to higher total transmit power, while $\bar R_s$ benefits additionally from the interference caused by residual AN at Eve. In the region of low $\beta$, the relationship $0 < \bar R_s^a < {{\bar R}_s}$ holds, thereby confirming the conclusions of Corollary 11.
\end{example}

\section{Conclusion}

In this paper, we have studied the performance of PLS using AN in situations where Eve employs advanced ANE to mitigate the impact of the AN. In particular, we have redefined the scaling laws for both average and instantaneous secrecy rates in the context of AN versus ANE. The resulting corollaries provide insights into interdependencies among the numbers of antennas at Alice, Bob, and Eve. Moreover, by comparing the cases where ANE counters AN with those where AN is not used, we have identified conditions under which AN remains effective. It is hoped that the analysis provided here will provide design guidance for secure communication when facing ANE countermeasures.

\section*{Appendix A\\
Proof of Theorem 1}

By definition, we have
\begin{equation}
{\bar R_s} = {{\rm E}_{\mathbf{H}}}\left[ {I\left( {{\mathbf{s}};{\mathbf{y}}\left| {\mathbf{H}} \right.} \right)} \right] - {{\rm E}_{{\mathbf{H}},{\mathbf{G}}}}\left[ {I\left( {{\mathbf{s}};{\mathbf{z}}\left| {{\mathbf{H}},{\mathbf{G}}} \right.} \right)} \right],
\end{equation}
where ${\mathbf{H}}$ and ${\mathbf{G}}$ are treated as Gaussian random matrices.

On writing Eq. (\ref{svd:1}) as ${\mathbf{H}} = {\mathbf{U\Lambda V}}_1^H$, we obtain
\begin{equation}
\begin{gathered}
\begin{aligned}
  {\mathbf{H}}{{\mathbf{V}}_1}{\mathbf{V}}_1^H{{\mathbf{H}}^H} &= \left( {{\mathbf{U\Lambda V}}_1^H} \right){{\mathbf{V}}_1}{\mathbf{V}}_1^H{\left( {{\mathbf{U\Lambda V}}_1^H} \right)^H} \hfill \\
   &= {\mathbf{U\Lambda V}}_1^H{{\mathbf{V}}_1}{{\mathbf{\Lambda }}^H}{{\mathbf{U}}^H} = {\mathbf{H}}{{\mathbf{H}}^H}. \hfill \\
\end{aligned}
\end{gathered}
\end{equation}

From (56) in \cite{ANE2} and (8) in \cite{SR1}, we arrive at
\begin{equation}\label{isy:1}
  \begin{gathered}
  \begin{aligned}
  {{\rm E}_{\mathbf{H}}}\left[ {I\left( {{\mathbf{s}};{\mathbf{y}}\left| {\mathbf{H}} \right.} \right)} \right] &= {{\rm E}_{\mathbf{H}}}\left( {\log \left| {{{\mathbf{I}}_{{N_b}}} + \alpha \gamma {\mathbf{H}}{{\mathbf{V}}_1}{\mathbf{V}}_1^H{{\mathbf{H}}^H}} \right|} \right) \hfill \\
   &= {{\rm E}_{\mathbf{H}}}\left( {\log \left| {{{\mathbf{I}}_{{N_b}}} + \alpha \gamma {\mathbf{H}}{{\mathbf{H}}^H}} \right|} \right) \hfill \\
   &= f\left( {{N_a},{N_b},\alpha \gamma {N_a} } \right), \hfill \\
   \end{aligned}
\end{gathered}
\end{equation}
where $f\left( {t,r,x} \right)$ is defined in (\ref{func:1}).

As inferred, the objective function of (\ref{W:1}) is non-negative. In the case of ${N_e} > {N_a} - {N_b}$, a minimum zero objective function ${\text{tr}}\left( {{\mathbf{WG}}{{\mathbf{V}}_0}{\mathbf{V}}_0^H{{\mathbf{G}}^H}{{\mathbf{W}}^H}} \right)=0$ can be achieved by applying the optimal projection matrix
\begin{equation}\label{optW:1}
{\mathbf{W}} = {\mathbf{V}}_2^H \in {\mathbb{C}^{\left( {{N_e} - {N_a} + {N_b}} \right) \times {N_e}}},
\end{equation}
where ${{{\mathbf{V}}_2}}$ is the null space of ${\mathbf{V}}_0^H{{\mathbf{G}}^H}$ that can be obtained by SVD, i.e.,
\begin{equation}
{\mathbf{V}}_0^H{{\mathbf{G}}^H} = {{\mathbf{U}}_2}\left[ {\begin{array}{*{20}{c}}
  {{{\mathbf{\Lambda }}_2}}&{\mathbf{0}}
\end{array}} \right]{\left[ {\begin{array}{*{20}{c}}
  {{{\mathbf{V}}_3}}&{{{\mathbf{V}}_2}}
\end{array}} \right]^H}.
\end{equation}

The optimal projection matrix in (\ref{optW:1}) ensures a perfect elimination of AN, and thus we have
\begin{equation}\label{isz:1}
{{\rm E}_{{\mathbf{H}},{\mathbf{G}}}}\left[ {I\left( {{\mathbf{s}};{\mathbf{z}}\left| {{\mathbf{H}},{\mathbf{G}}} \right.} \right)} \right] = {{\rm E}_{{{\mathbf{G}}_1}}}\left( {\log \left| {{{\mathbf{I}}_{{N_e} - {N_a} + {N_b}}} + \alpha {{\mathbf{G}}_1}{\mathbf{G}}_1^H} \right|} \right),
\end{equation}
where ${{\mathbf{G}}_1} = {\mathbf{V}}_2^H{\mathbf{G}}{{\mathbf{V}}_1} \in {\mathbb{C}^{\left( {{N_e} - {N_a} + {N_b}} \right) \times {N_b}}}$ is complex Gaussian random matrices with i.i.d. entries $ \sim \mathcal{C}\mathcal{N}(0,1)$.

According to (56) in \cite{ANE2}, the expression of (\ref{isz:1}) can be simplified to
\begin{equation}\label{isz:2}
{{\rm E}_{{\mathbf{H}},{\mathbf{G}}}}\left[ {I\left( {{\mathbf{s}};{\mathbf{z}}\left| {{\mathbf{H}},{\mathbf{G}}} \right.} \right)} \right] = f\left( {{N_b},{N_e} - {N_a} + {N_b},\alpha {N_b}} \right).
\end{equation}

Substituting (\ref{isy:1}) and (\ref{isz:2}) into (\ref{vr:1}), the proof of Theorem 1 is completed.

\section*{Appendix B\\
Proof of Theorem 2}

In the case of ${N_e} \leqslant {N_a} - {N_b}$, the objective function of (\ref{W:1}) cannot reach zero via (\ref{optW:1}). Instead, its minimum value corresponds to the minimum eigenvalue of ${\mathbf{V}}_0^H{{\mathbf{G}}^H}$ when the optimal projection vector ${{\mathbf{w}}^H} \in {\mathbb{C}^{1 \times {N_e}}}$ satisfies
\begin{equation}\label{optw:2}
{\mathbf{V}}_0^H{{\mathbf{G}}^H}{\mathbf{w}} = {\lambda _{\min }}{\mathbf{w}},
\end{equation}
where ${\lambda _{\min }}$ denotes the minimum eigenvalue of ${\mathbf{V}}_0^H{{\mathbf{G}}^H}$ and ${\mathbf{w}}$ represents the eigenvector corresponding to the minimum eigenvalue.

Based on the optimal projection vector in (\ref{optw:2}), the average achievable rate of Eve is expressed as
\begin{equation}\label{aare:1}
{{\rm E}_{{\mathbf{H}},{\mathbf{G}}}}\left[ {I\left( {{\mathbf{s}};{\mathbf{z}}\left| {{\mathbf{H}},{\mathbf{G}}} \right.} \right)} \right] = {{\text{E}}_{\mathbf{g}}}\left( {\log \left( {1 + \frac{{\sigma _s^2{{\mathbf{g}}^H}{\mathbf{g}}}}{{\sigma _r^2\lambda _{\min }^2 + \sigma _n^2}}} \right)} \right),
\end{equation}
where ${{\mathbf{g}}^H} = {{\mathbf{w}}^H}{\mathbf{G}}{{\mathbf{V}}_1} \in {\mathbb{C}^{{N_b}}}$ denotes the equivalent wiretap channel through which the information-bearing signal passes.
{Since the wiretap channel ${\mathbf{G}}$ is independent of the legitimate channel ${\mathbf{H}}$ and its subspace basis ${{\mathbf{V}}_1}$, the product ${\mathbf{G}}{{\mathbf{V}}_1}$ remains a Gaussian random matrix. Moreover, ${\mathbf{G}}{{\mathbf{V}}_0}$ and ${\mathbf{G}}{{\mathbf{V}}_1}$ are statistically independent due to the orthogonality of ${{\mathbf{V}}_0}$ and ${{\mathbf{V}}_1}$, as shown in (57) of \cite{ANE2}. Since ${\mathbf{w}}$ is a unit eigenvector of ${\mathbf{G}}{{\mathbf{V}}_0}$ and thus independent of ${\mathbf{G}}{{\mathbf{V}}_1}$, the equivalent wiretap channel ${{\mathbf{g}}^H}$ forms a Gaussian random vector.}

We denote $A$ as the square of the norm of ${\mathbf{g}}$, then we have
\begin{equation}\label{ll:1}
A \triangleq {{\mathbf{g}}^H}{\mathbf{g}} = \sum\limits_{i = 1}^{{N_b}} {{{\left| {{g_i}} \right|}^2}}  = \sum\limits_{i = 1}^{{N_b}} {\operatorname{Re} {{\left( {{g_i}} \right)}^2}}  + \sum\limits_{i = 1}^{{N_b}} {\operatorname{Im} {{\left( {{g_i}} \right)}^2}},
\end{equation}
where ${g_i}$ represents the $i$-th element of ${\mathbf{g}}$, following i.i.d. $ \mathcal{C}\mathcal{N}(0,1)$. Note that both its real and imaginary parts are independently distributed as $ \mathcal{C}\mathcal{N}(0,1/2)$, the variable $A$ follows a scaled chi-squared distribution with $2{N_b}$ DoFs, i.e., $A \sim {\chi ^2}(2{N_b})/2$. Therefore, the PDF of $A$ can be expressed as
\begin{equation}\label{chi:1}
{f_A}\left( a \right) = \frac{{{a^{{N_b} - 1}}{e^{ - a}}}}{{\Gamma \left( {{N_b},0} \right)}}.
\end{equation}

Additionally, we denote $B$ as the square of the minimum eigenvalue, then we have
\begin{equation}
B \triangleq \lambda _{\min }^2 = {\lambda _{\min }}\left( {{\mathbf{G}}{{\mathbf{V}}_0}{\mathbf{V}}_0^H{{\mathbf{G}}^H}} \right),
\end{equation}
where ${\lambda _{\min }}\left( {{\mathbf{G}}{{\mathbf{V}}_0}{\mathbf{V}}_0^H{{\mathbf{G}}^H}} \right)$ denotes the minimum eigenvalue of the matrix ${{\mathbf{G}}{{\mathbf{V}}_0}{\mathbf{V}}_0^H{{\mathbf{G}}^H}}$.

{Since the matrix ${{\mathbf{G}}{{\mathbf{V}}_0}{\mathbf{V}}_0^H{{\mathbf{G}}^H}}$ follows a Wishart distribution, the PDF of its minimum eigenvalue is given by \cite{Wishart1}, as follows:
\begin{equation}\label{Wishart:1}
{f_B}\left( b \right) = K\sum\limits_{n = 1}^{{N_e}} {\sum\limits_{m = 1}^{{N_e}} {{{\left( { - 1} \right)}^{n + m}}{b^{n + m - 2 + {N_a} - {N_b} - {N_e}}}{e^{ - b}}} } \left| {\mathbf{\Omega }} \right|,
\end{equation}
where the constant $K$ is given by
\begin{equation}
K = {\left[ {\prod _{i = 1}^{{N_e}}\left( {{N_a} - {N_b} - i} \right)!\prod _{j = 1}^{{N_e}}\left( {{N_e} - j} \right)!} \right]^{ - 1}},
\end{equation}
and the $\left( {i,j} \right)$-th element of ${\mathbf{\Omega }}$ is given by
\begin{equation}\label{ll:3}
{\Omega _{i,j}} = \Gamma \left( {\alpha _{i,j}^{\left( n \right)\left( m \right)} + {N_a} - {N_b} - {N_e} + 1,b} \right),
\end{equation}
with
\begin{equation}
\alpha _{i,j}^{\left( n \right)\left( m \right)} = \left\{ \begin{gathered}
  i + j - 2,{\text{   if }}i < n{\text{ and }}j < m, \hfill \\
  i + j,{\text{ \ \ \ \ \       if }}i \geqslant n{\text{ and }}j \geqslant m, \hfill \\
  i + j - 1,{\text{    otherwise}}{\text{.}} \ \ \ \ \ \ \ \ \ \ \ \\
\end{gathered}  \right.
\end{equation}}

Combining (\ref{chi:1}) and (\ref{Wishart:1}), the achievable rate of Eve can be written as
\begin{equation}\label{isz:3}
\begin{gathered}
  {{\rm E}_{{\mathbf{H}},{\mathbf{G}}}}\left[ {I\left( {{\mathbf{s}};{\mathbf{z}}\left| {{\mathbf{H}},{\mathbf{G}}} \right.} \right)} \right] \hfill \\
   = \int_0^\infty  {\int_0^\infty  {\log \left( {1 + \frac{{\alpha A}}{{\alpha \beta B + 1}}} \right){f_A}\left( a \right){f_B}\left( b \right){{\text{d}}a}{{\text{d}}b}} } . \hfill \\
\end{gathered}
\end{equation}

{Note that although both $A$ and $B$ are functions of the same channel matrix $\mathbf{G}$, their statistical relationship can be decoupled. Specifically, the quantity $B$ and the combining vector $\mathbf{w}$ are fully determined by the projected channel $\mathbf{G}\mathbf{V}_0$, while $A$ depends on $\mathbf{G}\mathbf{V}_1$ through the equivalent channel $\mathbf{g}=\mathbf{w}^H\mathbf{G}\mathbf{V}_1$. Due to the orthogonality between $\mathbf{V}_0$ and $\mathbf{V}_1$, the projected matrices $\mathbf{G}\mathbf{V}_0$ and $\mathbf{G}\mathbf{V}_1$ are statistically independent under the i.i.d. Rayleigh fading assumption (see (57) in \cite{ANE2}).

Although the vector $\mathbf{w}$ is correlated with $B$, this correlation only affects the direction of the projection applied to $\mathbf{G}\mathbf{V}_1$. Since $\mathbf{G}\mathbf{V}_1$ is isotropically distributed, its projection onto any unit-norm vector yields an identical statistical distribution. As a result, conditioned on $\mathbf{w}$, the distribution of $A=\|\mathbf{w}^H\mathbf{G}\mathbf{V}_1\|^2$ remains invariant and is independent of $B$.

This property allows the expectation in (\ref{isz:3}) to be evaluated via iterated averaging, i.e., by first averaging over $A$ conditioned on $\mathbf{w}$, and subsequently averaging over $B$. For notational simplicity, this procedure is expressed using the marginal PDFs of $A$ and $B$.}

Substituting (\ref{isy:1}) and (\ref{isz:3}) into (\ref{vr:1}), the proof of Theorem 2 is completed.

\section*{Appendix C\\
Proof of Theorem 3}

From (68) in \cite{ANE2}, we have
\begin{equation}\label{infi:1}
\frac{{I\left( {{\mathbf{s}};{\mathbf{y}}} \right)}}{{{N_b}}} \to \Phi \left( {{\alpha \gamma {N_b}},{\delta _2}} \right),
\end{equation}
where $\Phi \left( {x,y} \right)$ is given in (\ref{fidef:1}) and $I\left( {{\mathbf{s}};{\mathbf{y}}} \right)$ can be found in part of (\ref{isy:1}) as
\begin{equation}
I\left( {{\mathbf{s}};{\mathbf{y}}} \right) = \log \left| {{{\mathbf{I}}_{{N_b}}} + \alpha \gamma {\mathbf{H}}{{\mathbf{H}}^H}} \right|.
\end{equation}

As ${N_a},{N_b},{N_e} \to \infty $ with ${N_e}/{N_b} \to {\delta _1}$, ${N_a}/{N_b} \to {\delta _2}$, and ${\delta _1} > {\delta _2} - 1$, we arrive at
\begin{equation}
\begin{gathered}
  \frac{{I\left( {{\mathbf{s}};{\mathbf{z}}} \right)}}{{{N_e} - {N_a} + {N_b}}} \hfill \\
   \to \Phi \left( {\alpha \left( {{N_e} - {N_a} + {N_b}} \right),\frac{{{N_b}}}{{{N_e} - {N_a} + {N_b}}}} \right) \hfill \\
   = \Phi \left( {\frac{{\alpha {N_b}\left( {{N_e} - {N_a} + {N_b}} \right)}}{{{N_b}}},\frac{{{N_b}}}{{{N_e} - {N_a} + {N_b}}}} \right) \hfill \\
   = \Phi \left( {\alpha {N_b}\left( {{\delta _1} - {\delta _2} + 1} \right),\frac{1}{{{\delta _1} - {\delta _2} + 1}}} \right), \hfill \\
\end{gathered}
\end{equation}
where $I\left( {{\mathbf{s}};{\mathbf{y}}} \right)$ can be found in part of (\ref{isz:1}) as
\begin{equation}\label{ISZ:1}
I\left( {{\mathbf{s}};{\mathbf{z}}} \right) = \log \left| {{{\mathbf{I}}_{{N_e} - {N_a} + {N_b}}} + \alpha {{\mathbf{G}}_1}{\mathbf{G}}_1^H} \right|.
\end{equation}

On defining ${\delta _3} = {\delta _1} - {\delta _2} + 1$, we get
\begin{equation}\label{infi:2}
\frac{{I\left( {{\mathbf{s}};{\mathbf{z}}} \right)}}{{{N_b}}} \to {\delta _3}\Phi \left( {\alpha {N_b}{\delta _3},\frac{1}{{{\delta _3}}}} \right).
\end{equation}

Since Eq. (\ref{Fdef:1}) reveals the property of $\mathcal{F}\left( {x,y} \right) = \mathcal{F}\left( {xy,{y^{ - 1}}} \right)$, it is straightforward to establish the following property:
\begin{equation}\label{phipro:2}
\Phi \left( {x,y} \right) = y\Phi \left( {xy,{y^{ - 1}}} \right).
\end{equation}

Using (\ref{phipro:2}), (\ref{infi:2}) can be simplified to
\begin{equation}\label{infi:3}
\frac{{I\left( {{\mathbf{s}};{\mathbf{z}}} \right)}}{{{N_b}}} \to \Phi \left( {\alpha {N_b},{\delta _3}} \right).
\end{equation}

Combining (\ref{infi:1}) and (\ref{infi:3}), the proof of Theorem 3 is completed.

\section*{Appendix D\\
Proof of Theorem 4}

Based on the optimal projection vector in (\ref{optw:2}), the instantaneous achievable rate of Eve can be expressed as
\begin{equation}\label{appd:1}
I\left( {{\mathbf{s}};{\mathbf{z}}} \right) = \log \left( {1 + \frac{{\sigma _s^2{{\mathbf{g}}^H}{\mathbf{g}}}}{{\sigma _r^2{\lambda _{{{\min }^2}}} + \sigma _n^2}}} \right).
\end{equation}

As ${N_a},{N_b},{N_e} \to \infty $ with ${N_e}/{N_b} \to {\delta _1}$, ${N_a}/{N_b} \to {\delta _2}$, and ${\delta _1}  \leqslant  {\delta _2} - 1$, we have $\lambda _{\min }^2 \to 0$. Since $\beta  < M$, the instantaneous achievable rate of Eve becomes
\begin{equation}
I\left( {{\mathbf{s}};{\mathbf{z}}} \right) = \log \left( {1 + \alpha {{\mathbf{g}}^H}{\mathbf{g}}} \right),
\end{equation}
where ${\mathbf{g}} \in {\mathbb{C}^{{N_b}}}$ is a complex Gaussian vector with the distribution of $\mathcal{C}\mathcal{N}({\mathbf{0}},{{\mathbf{I}}_{{N_b}}})$ when the wiretap channel ${\mathbf{G}}$ experiences Rayleigh fading. As ${N_b} \to \infty$, we arrive at
\begin{equation}\label{ARE:2}
\frac{{I\left( {{\mathbf{s}};{\mathbf{z}}} \right)}}{{{N_b}}} \to \mathop {\lim }\limits_{{N_b} \to \infty } \frac{1}{{{N_b}}}\log \left( {1 + \alpha {N_b}} \right) = 0.
\end{equation}

Combining (\ref{infi:1}) and (\ref{ARE:2}), the proof of Theorem 4 is completed.

\section*{Appendix E\\
Proof of Corollary 1}

The definition of $f\left( {t,r,x} \right)$ in (\ref{func:1}) exhibits the property of
\begin{equation}\label{pp:1}
f\left( {t,r,x} \right) = f\left( {r,t,x} \right).
\end{equation}

Meanwhile, the last equation of (\ref{isz:1}) reveals the following two properties. For ${{x_1}} \leqslant {{x_2}}$, it follows that
\begin{equation}\label{pp:2}
f\left( {t,r,{x_1}} \right)  \leqslant  f\left( {t,r,{x_2}} \right),
\end{equation}
since a higher SNR results in a larger achievable rate.

For ${{t_1}} \leqslant {{t_2}}$, it follows that
\begin{equation}\label{pp:3}
f\left( {{t_1},r,x} \right)  \leqslant  f\left( {{t_2},r,x} \right),
\end{equation}
as increasing the number of antennas enhances the achievable rate.

If $\gamma  \leqslant {N_b}/{N_a}$, then we have
\begin{equation}\label{scaling:1}
f\left( {{N_a},{N_b},\alpha \gamma {N_a}} \right) \leqslant f\left( {{N_a},{N_b},\alpha {N_b}} \right).
\end{equation}

If ${N_e} \geqslant 2{N_a} - {N_b}$, using (\ref{pp:1}) and (\ref{pp:3}), we have
\begin{equation}\label{scaling:2}
\begin{gathered}
  f\left( {{N_b},{N_e} - {N_a} + {N_b},\alpha {N_b}} \right) \hfill \\
   = f\left( {{N_e} - {N_a} + {N_b},{N_b},\alpha {N_b}} \right) \hfill \\
   \geqslant f\left( {{N_a},{N_b},\alpha {N_b}} \right). \hfill \\
\end{gathered}
\end{equation}

Substituting (\ref{scaling:1}) and (\ref{scaling:2}) into (\ref{theo:1}), the proof of Corollary 2 is completed.

\section*{Appendix F\\
Proof of Corollary 3}

Define ${\delta _3} = {\delta _1} - {\delta _2} + 1$, given a fixed ${{\delta _2}}$, as ${\delta _1} \to \infty $, we get ${\delta _3} \to \infty $.

{
Applying (\ref{phipro:2}) on the second term of (\ref{theo:3}), we arrive at
\begin{equation}\label{comment:1}
\mathop {\lim }\limits_{{\delta _3} \to \infty } \Phi \left( {\alpha {N_b},{\delta _3}} \right) = \mathop {\lim }\limits_{{\delta _3} \to \infty } {\delta _3}\Phi \left( {\alpha {N_b}{\delta _3},\frac{1}{{{\delta _3}}}} \right).
\end{equation}

Substituting the definition in (\ref{fidef:1}) into (\ref{comment:1}), the above limit can be expanded as
\begin{equation}\label{comment:2}
\begin{gathered}
  \mathop {\lim }\limits_{{\delta _3} \to \infty } {\delta _3}\Phi \left( {\alpha {N_b}{\delta _3},\frac{1}{{{\delta _3}}}} \right) \hfill \\
   = \mathop {\lim }\limits_{{\delta _3} \to \infty } \log \left( {1 + \alpha {N_b}{\delta _3} - \mathcal{F}\left( {\alpha {N_b}{\delta _3},\frac{1}{{{\delta _3}}}} \right)/4} \right) \hfill \\
   + {\delta _3}\log \left( {1 + \alpha {N_b} + \mathcal{F}\left( {\alpha {N_b}{\delta _3},\frac{1}{{{\delta _3}}}} \right)/4} \right) \hfill \\
   - \frac{{\log e}}{{4\alpha {N_b}}}\mathcal{F}\left( {\alpha {N_b}{\delta _3},\frac{1}{{{\delta _3}}}} \right). \hfill \\
\end{gathered}
\end{equation}

From the definition in (\ref{Fdef:1}), we further have
\begin{equation}\label{comment:3}
\begin{gathered}
  \mathop {\lim }\limits_{{\delta _3} \to \infty } \mathcal{F}\left( {\alpha {N_b}{\delta _3},\frac{1}{{{\delta _3}}}} \right) = \mathop {\lim }\limits_{{\delta _3} \to \infty } \mathcal{F}\left( {\alpha {N_b},{\delta _3}} \right) \hfill \\
   = \mathop {\lim }\limits_{{\delta _3} \to \infty } {\left( {\sqrt {\alpha {N_b}{{\left( {1 + \sqrt {{\delta _3}} } \right)}^2} + 1}  - \sqrt {\alpha {N_b}{{\left( {1 - \sqrt {{\delta _3}} } \right)}^2} + 1} } \right)^2} \hfill \\
   = \mathop {\lim }\limits_{{\delta _3} \to \infty } {\left( {\frac{{4\alpha {N_b}\sqrt {{\delta _3}} }}{{\sqrt {\alpha {N_b}{{\left( {1 + \sqrt {{\delta _3}} } \right)}^2} + 1}  + \sqrt {\alpha {N_b}{{\left( {1 - \sqrt {{\delta _3}} } \right)}^2} + 1} }}} \right)^2} \hfill \\
   = 4\alpha {N_b}. \hfill \\
\end{gathered}
\end{equation}

Substituting (\ref{comment:3}) into (\ref{comment:2}), we finally obtain
\begin{equation}
\begin{gathered}
  \mathop {\lim }\limits_{{\delta _3} \to \infty } {\delta _3}\Phi \left( {\alpha {N_b}{\delta _3},\frac{1}{{{\delta _3}}}} \right) \hfill \\
   = \mathop {\lim }\limits_{{\delta _3} \to \infty } \log \left( {1 + \alpha {N_b}\left( {{\delta _3} - 1} \right)} \right) \hfill \\
   + {\delta _3}\log \left( {1 + 2\alpha {N_b}} \right) - \log e \to \infty . \hfill \\
\end{gathered}
\end{equation}
}

Given fixed ${{P_s}}$ and ${{\delta _2}}$, it follows
\begin{equation}
\mathop {\lim }\limits_{{\delta _1} \to \infty } \Phi \left( {\alpha \gamma {N_b},{\delta _2}} \right) - \Phi \left( {\alpha {N_b},{\delta _3}} \right) < 0.
\end{equation}

According to the definition of ${R_s}$ in (\ref{vr:1}), ${R_s}/{N_b} \to 0$ is proved.

\section*{Appendix G\\
Proof of Corollary 5}

When ${N_e} > {N_a} - {N_b}$ is established, the approximations for average and instantaneous secrecy rates are provided by (\ref{theo:3}), expressed as
\begin{equation}\label{appro:1}
R_s^{app} = {{N_b}{{\left[ {\Phi \left( {\alpha \gamma {N_b},\frac{{{N_a}}}{{{N_b}}}} \right) - \Phi \left( {\alpha {N_b},\frac{{{N_e} - {N_a} + {N_b}}}{{{N_b}}}} \right)} \right]}^ + }}.
\end{equation}

When it comes to $ {N_e} \leqslant {N_a} - {N_b}$, the instantaneous achievable rate of Eve is given by (\ref{appd:1}), shown as
\begin{equation}\label{appro:2}
I\left( {{\mathbf{s}};{\mathbf{z}}} \right) = \log \left( {1 + \frac{{\alpha {{\mathbf{g}}^H}{\mathbf{g}}}}{{\alpha \beta \lambda _{\min }^2 + 1}}} \right).
\end{equation}

Since ${{N_a}}$, ${{N_b}}$, ${{N_e}}$ are finite, the result $\lambda _{\min }^2 \to 0$ no longer holds. Based on the PDF provided in (\ref{Wishart:1}), $\lambda _{\min }^2$ can be approximated by its expectation as
\begin{equation}
\mu  = \int_0^\infty  {B{f_B}\left( b \right){\text{d}}b}.
\end{equation}

According to (\ref{infi:1}) and (\ref{appro:2}), the approximations for average and instantaneous secrecy rates are written as
\begin{equation}\label{appro:3}
{R_s^{app} = {N_b}\Phi \left( {\alpha \gamma {N_b},\frac{{{N_a}}}{{{N_b}}}} \right) - \log \left( {1 + \frac{{\alpha {N_b}}}{{\alpha \beta \mu  + 1}}} \right){\text{.}}}
\end{equation}

Combining (\ref{appro:1}) and (\ref{appro:3}), the proof of Corollary 5 is completed.

\section*{Appendix H\\
Proof of Corollary 6}

Based on (\ref{aare:1}), the average achievable rate of Eve can be rewritten as
\begin{equation}
{{\text{E}}_{{\mathbf{H}},{\mathbf{G}}}}\left[ {I\left( {{\mathbf{s}};{\mathbf{z}}\left| {{\mathbf{H}},{\mathbf{G}}} \right.} \right)} \right] = {{\rm E}_{\mathbf{g}}}\left( {\log \left( {1 + \frac{{\alpha {{\mathbf{g}}^H}{\mathbf{g}}}}{{\alpha \beta \lambda _{\min }^2 + 1}}} \right)} \right).
\end{equation}

In the case of ${N_e} \leqslant {N_a} - {N_b}$, the matrix ${\mathbf{G}}{{\mathbf{V}}_0}{\mathbf{V}}_0^H{{\mathbf{G}}^H}$ has full rank, indicating that ${\lambda _{\min }} \neq 0$. As $\beta  \to \infty $, we arrive at
\begin{equation}\label{betainf:1}
\mathop {\lim }\limits_{\beta  \to \infty } {{\text{E}}_{{\mathbf{H}},{\mathbf{G}}}}\left[ {I\left( {{\mathbf{s}};{\mathbf{z}}\left| {{\mathbf{H}},{\mathbf{G}}} \right.} \right)} \right] = 0.
\end{equation}

Substituting (\ref{betainf:1}) into (\ref{vr:2}), we have
\begin{equation}\label{rscb:1}
{\bar R_s} = {{\text{E}}_{\mathbf{H}}}\left[ {I\left( {{\mathbf{s}};{\mathbf{y}}\left| {\mathbf{H}} \right.} \right)} \right] = {\bar C_b},
\end{equation}
where ${\bar C_b}$ denotes the average channel capacity for Bob, and can be achieved  the input $\mathbf{s}$ is a complex Gaussian random vector with the distribution of $\mathcal{C}\mathcal{N}({\mathbf{0}},\sigma _s^2{{\mathbf{I}}_{{N_b}}})$ (see Theorem 1 in \cite{CC1}).

Eq. (\ref{cs:1}) leads to the fact that
\begin{equation}\label{rscb:2}
{\bar R_s} \leqslant {\bar C_s} \leqslant {\bar C_b}.
\end{equation}

The combination of (\ref{rscb:1}) and (\ref{rscb:2}) results in
\begin{equation}
{\bar R_s} = {\bar C_s} = {\bar C_b},
\end{equation}
as $\beta  \to \infty $ in the case of ${N_e} \leqslant {N_a} - {N_b}$.

\section*{Appendix I\\
Proof of Theorem 5}

Without using AN, Alice transmits only information-bearing signals, and to ensure the signal power remains normalized, the total transmission power is
\begin{equation}\label{Pt:1}
P = {P_s} + {P_r} = \alpha \gamma \sigma _u^2\left[ {{N_b} + \beta \left( {{N_a} - {N_b}} \right)} \right].
\end{equation}

Similar to (\ref{isy:1}), the average achievable rate at Bob can be expressed as
\begin{equation}\label{noansc:1}
\begin{gathered}
  \bar R_b^a = {{\text{E}}_{\mathbf{H}}}\left[ {I\left( {{\mathbf{s}};{\mathbf{y}}\left| {\mathbf{H}} \right.} \right)} \right] \hfill \\
  \ \ \ \ = {{\text{E}}_{\mathbf{H}}}\left( {\log \left| {{{\mathbf{I}}_{{N_b}}} + \alpha \gamma \left( {1 + \frac{{\beta \left( {{N_a} - {N_b}} \right)}}{{{N_b}}}} \right){\mathbf{H}}{{\mathbf{H}}^H}} \right|} \right) \hfill \\
  \ \ \ \ = f\left( {{N_a},{N_b},\alpha \gamma {N_a}\left( {1 + \frac{{\beta \left( {{N_a} - {N_b}} \right)}}{{{N_b}}}} \right)} \right). \hfill \\
\end{gathered}
\end{equation}

Since Eve is not disturbed by AN, the average achievable rate at Eve can be given by
\begin{equation}\label{noansc:2}
\begin{gathered}
  \bar R_e^a = {{\text{E}}_{\mathbf{G}}}\left[ {I\left( {{\mathbf{s}};{\mathbf{y}}\left| {\mathbf{G}} \right.} \right)} \right] \hfill \\
  \ \ \ \  = {{\text{E}}_{\mathbf{G}}}\left( {\log \left| {{{\mathbf{I}}_{{N_e}}} + \alpha \left( {1 + \frac{{\beta \left( {{N_a} - {N_b}} \right)}}{{{N_b}}}} \right){\mathbf{G}}{{\mathbf{V}}_1}{\mathbf{V}}_1^H{{\mathbf{G}}^H}} \right|} \right) \hfill \\
  \ \ \ \  = f\left( {{N_a},{N_e},\alpha {N_a}\left( {1 + \frac{{\beta \left( {{N_a} - {N_b}} \right)}}{{{N_b}}}} \right)} \right). \hfill \\
\end{gathered}
\end{equation}

Substituting (\ref{noansc:1}) and (\ref{noansc:2}) into (\ref{vr:1}), the proof of Theorem 5 is completed.

\section*{Appendix J\\
Proof of Theorem 6}

{Without using AN, Alice transmits only information-bearing signals, and the total transmission power is given in (\ref{Pt:1}). As such, the instantaneous achievable rate at Bob can be expressed as
\begin{equation}
R_b^a = I\left( {{\mathbf{s}};{\mathbf{y}}\left| {\mathbf{H}} \right.} \right)\; = \log \left| {{{\mathbf{I}}_{{N_b}}} + \frac{\eta }{{{N_a}}}{\mathbf{H}}{{\mathbf{H}}^H}} \right|,
\end{equation}
where $\eta $ is defined in (\ref{theo:5}).

As ${N_a},{N_b} \to \infty $, from (68) in \cite{ANE2}, we have
\begin{equation}\label{arba:1}
\frac{{R_b^a}}{{{N_b}}} \to \Phi \left( {\frac{\eta }{{{\delta _2}}},{\delta _2}} \right).
\end{equation}

Since Eve is not affected by AN, the instantaneous achievable rate at Eve can be given by
\begin{equation}
R_e^a = I\left( {{\mathbf{s}};{\mathbf{y}}\left| {\mathbf{G}} \right.} \right) = \log \left| {{{\mathbf{I}}_{{N_e}}} + \frac{\eta }{{\gamma {N_a}}}{\mathbf{G}}{{\mathbf{V}}_1}{\mathbf{V}}_1^H{{\mathbf{G}}^H}} \right|.
\end{equation}

Similarly, as ${N_a},{N_e} \to \infty $, we obtain
\begin{equation}\label{area:1}
\frac{{R_e^a}}{{{N_e}}} \to \Phi \left( {\frac{{\eta {\delta _1}}}{{\gamma {\delta _2}}},\frac{{{\delta _2}}}{{{\delta _1}}}} \right).
\end{equation}

Using (\ref{area:1}), we arrive at
\begin{equation}\label{area:2}
\frac{{R_e^a}}{{{N_b}}} = \frac{{{N_e}}}{{{N_b}}}\frac{{R_e^a}}{{{N_e}}} \to {\delta _1}\Phi \left( {\frac{{\eta {\delta _1}}}{{\gamma {\delta _2}}},\frac{{{\delta _2}}}{{{\delta _1}}}} \right).
\end{equation}

Combining (\ref{arba:1}) and (\ref{area:2}), the proof of Theorem 6 is completed.}

\section*{Appendix K\\
Proof of Corollary 10}

According to Corollary 7, $\gamma  \leqslant 1$ and ${N_e} \geqslant {N_b}$ lead to $\bar R_s^a = 0$.

According to Theorem 1, ${N_a} - {N_b} < {N_e} < 2{N_a} - {N_b}$ and $\gamma  > {N_b}/{N_a}$ result in $\bar R_s > 0$.

Due to the increase of ${N_e}$ reduces the value of $\bar R_s$, ${N_e} \leqslant {N_a} - {N_b}$ also leads to $\bar R_s > 0$.

In summary, ${N_b}/{N_a} < \gamma  \leqslant 1$ combined with ${N_b} \leqslant {N_e} < 2{N_a} - {N_b}$ result in $0 = \bar R_s^a < {{\bar R}_s}$.

\section*{Appendix L\\
Proof of Corollary 11}

According to (\ref{noansc:1}), as $\beta  \to 0$, the average achievable rate at Bob without the AN scheme is expressed as
\begin{equation}
\bar R_b^a = {{\rm E}_{\mathbf{H}}}\left[ {I\left( {{\mathbf{s}};{\mathbf{y}}\left| {\mathbf{H}} \right.} \right)} \right] = {{\text{E}}_{\mathbf{H}}}\left( {\log \left| {{{\mathbf{I}}_{{N_b}}} + \alpha \gamma {\mathbf{H}}{{\mathbf{H}}^H}} \right|} \right),
\end{equation}
which is the same as the one with the AN scheme $\bar R_b$ shown in (\ref{isy:1}).

According to (\ref{noansc:2}), as $\beta  \to 0$, the average achievable rate at Eve without the AN scheme is expressed as
\begin{equation}
\bar R_e^a = {{\rm E}_{\mathbf{G}}}\left[ {I\left( {{\mathbf{s}};{\mathbf{y}}\left| {\mathbf{G}} \right.} \right)} \right] = {{\text{E}}_{\mathbf{G}}}\left( {\log \left| {{{\mathbf{I}}_{{N_e}}} + \alpha {\mathbf{G}}{{\mathbf{V}}_1}{\mathbf{V}}_1^H{{\mathbf{G}}^H}} \right|} \right).
\end{equation}

According to (\ref{isz:1}), if ${N_e} > {N_a} - {N_b}$ is established, the average achievable rate at Eve with the AN scheme is shown as
\begin{equation}
\begin{gathered}
  {{\bar R}_e} = {{\text{E}}_{{\mathbf{H}},{\mathbf{G}}}}\left[ {I\left( {{\mathbf{s}};{\mathbf{z}}\left| {{\mathbf{H}},{\mathbf{G}}} \right.} \right)} \right] \hfill \\
  \ \ \ \ = {{\text{E}}_{{{\mathbf{G}}_1}}}\left( {\log \left| {{{\mathbf{I}}_{{N_e} - {N_a} + {N_b}}} + \alpha {{\mathbf{G}}_1}{\mathbf{G}}_1^H} \right|} \right). \hfill \\
\end{gathered}
\end{equation}

Since both ${{\mathbf{G}}{{\mathbf{V}}_1}}$ and ${{\mathbf{G}}_1} = {\mathbf{V}}_2^H{\mathbf{G}}{{\mathbf{V}}_1}$ are complex
Gaussian random matrices with i.i.d. entries distributed as $ \sim \mathcal{C}\mathcal{N}(0,1)$, the inequality ${{\bar R}_e} < \bar R_e^a$ can be obtained by applying the requirement of AN generation ${{N_a}} > {{N_b}}$, indicating that an increased number of antennas at the eavesdropper enhances Eve's ability to intercept information.

According to (\ref{aare:1}), if ${N_e}  \leqslant  {N_a} - {N_b}$ is established, the average achievable rate at Eve with the AN scheme is also upper bounded by
\begin{equation}
\begin{gathered}
  {{\bar R}_e} = {{\text{E}}_{{\mathbf{H}},{\mathbf{G}}}}\left[ {I\left( {{\mathbf{s}};{\mathbf{z}}\left| {{\mathbf{H}},{\mathbf{G}}} \right.} \right)} \right] \hfill \\
  \ \ \ \ = {{\text{E}}_{\mathbf{g}}}\left( {\log \left( {1 + \frac{{\sigma _s^2{{\mathbf{g}}^H}{\mathbf{g}}}}{{\sigma _r^2{\lambda _{{{\min }^2}}} + \sigma _n^2}}} \right)} \right) \hfill \\
  \ \ \ \ < {{\text{E}}_{\mathbf{g}}}\left( {\log \left( {1 + \alpha {{\mathbf{g}}^H}{\mathbf{g}}} \right)} \right) < \bar R_e^a. \hfill \\
\end{gathered}
\end{equation}

Based on (\ref{theo:5}), ${N_e} < {N_b}$ and $\gamma>1$ lead to $\bar R_s^a > 0$. Therefore, according to the definition of (\ref{vr:2}), we arrive at
\begin{equation}
0 < \bar R_s^a < {{\bar R}_s}.
\end{equation}
\ifCLASSOPTIONcaptionsoff
  \newpage
\fi


\begin{thebibliography}{00}
\bibitem{CM1}
W. Diffie and M. E. Hellman, ``New directions in cryptography,'' \emph{IEEE Trans. Inf. Theory}, vol. 22, no. 6, pp. 644-654, Nov. 1976.
\bibitem{CM2}
M. Hellman, ``An extension of the Shannon theory approach to cryptography,'' \emph{IEEE Trans. Inf. Theory}, vol. 23, no. 3, pp. 289-294, May 1977.
\bibitem{CM3}
R. Ahlswede and I. Csisz\'{a}r, ``Common randomness in information theory and cryptography. I. Secret sharing,'' \emph{IEEE Trans. Inf. Theory},, vol. 39, no. 4, pp. 1121-1132, Jul. 1993.
\bibitem{PLS1}
Mukherjee, S. A. A. Fakoorian, J. Huang, and A. L. Swindlehurst, ``Principles of physical layer security in multiuser wireless networks: A survey,'' \emph{IEEE Commun. Surv. Tut.}, vol. 16, no. 3, pp. 1550-1573, 3rd Quart. 2014.
\bibitem{PLS2}
D. Wang, B. Bai, W. Zhao, and Z. Han, ``A survey of optimization approaches for wireless physical layer security,'' \emph{IEEE Commun. Surv. Tut.}, vol. 21, no. 2, pp. 1878-1911, 2nd Quart. 2019.
\bibitem{PLS3}
Y. Zou, J. Zhu, X. Wang, and L. Hanzo, ``A survey on wireless security: Technical challenges, recent Advances, and future trends,'' \emph{Proc. IEEE Inst. Electr. Electron. Eng.}, vol. 104, no. 9, pp. 1727-1765, Sep. 2016.
\bibitem{PLS4}
J. M. Hamamreh, H. M. Furqan, and H. Arslan, ``Classifications and applications of physical layer security techniques for confidentiality: A comprehensive survey,'' \emph{IEEE Commun. Surv. Tut.}, vol. 21, no. 2, pp. 1773-1828, 2nd Quart. 2019.
\bibitem{PLS6}
H. Niu, \emph{et al.}, `` A survey on artificial noise for physical layer security: Opportunities, technologies, guidelines, advances, and trends,'' \emph{IEEE Commun. Surv. Tut.}, vol. 28, pp. 341-381, Sep. 2025.
\bibitem{PLS7}
J. Chen \emph{et al.}, ``A survey on directional modulation: Opportunities, challenges, recent advances, implementations, and future trends,'' \emph{IEEE Internet Things J.}, vol. 12, no. 16, pp. 32581-32615, Aug. 2025.
\bibitem{PLS5}
F. Naeem, M. Ali, G. Kaddoum, C. Huang, and C. Yuen, ``Security and privacy for reconfigurable intelligent surface in 6G: A review of prospective applications and challenges,'' \emph{IEEE Open J. Commun. Soc.}, vol. 4, pp. 1196-1217, May 2023.
\bibitem{wyner1}
A. D. Wyner, ``The wire-tap channel,'' \emph{Bell Syst. Tech. J.}, vol. 54, no. 8, pp. 1355-1387, Oct. 1975.
\bibitem{AN1}
R. Negi and S. Goel, ``Secret communication using artificial noise,'' in \emph{Proc. IEEE 62nd Veh. Technol. Conf. (VTC-Fall)}, Dallas, TX, USA, Sep. 2005, pp. 1906-1910.
\bibitem{AN2}
S. Goel and R. Negi, ``Guaranteeing secrecy using artificial noise,'' \emph{IEEE Trans. Wireless Commun.}, vol. 7, no. 6, pp. 2180-2189, Jun. 2008.
\bibitem{AN3}
Z. Wang, M. Xiao, M. Skoglund, and H. V. Poor, ``Secure degrees of freedom of wireless X networks using artificial noise alignment,'' \emph{IEEE Trans. Commun.}, vol. 63, no. 7, pp. 2632-2646, Jul. 2015.
\bibitem{AN4}
S. Lin, K. Huang, W. Luo, and Y. Zou, ``Analysis of pilot contamination on the security performance of artificial noise in MIMO systems,'' in \emph{Proc. IEEE 81st Veh. Technol. Conf. (VTC-Spring)}, Glasgow, UK, May. 2015, pp. 1-5.
\bibitem{AN5}
N. Nguyen, H. Q. Ngo, T. Q. Duong, H. D. Tuan, and K. Tourki, ``Secure massive MIMO with the artificial noise-aided downlink training,'' \emph{IEEE J. Sel. Areas Commun.}, vol. 36, no. 4, pp. 802-816, Apr. 2018.
\bibitem{AN6}
L. Sun, L. Cao, Z. Tang, and Y. Feng,``Artificial-noise-aided secure multi-user multi-antenna transmission with quantized CSIT: a comprehensive design and analysis,'' \emph{IEEE Trans. Inf. Forensics and Secur.}, vol. 15, pp. 3734-3748, 2020.
\bibitem{AN7}
W. Wang, K. C. Teh, and K. H. Li, ``Artificial noise aided physical layer security in multi-antenna small-cell networks,'' \emph{IEEE Trans. Inf. Forensics Secur.}, vol. 12, no. 6, pp. 1470-1482, Jun. 2017.
\bibitem{AN9}
C. Song, ``Leakage rate analysis for artificial noise assisted massive MIMO with non-coherent passive eavesdropper in block-fading,'' \emph{IEEE Trans. Wireless Commun.}, vol. 18, no. 4, pp. 2111-2124, Apr. 2019.
\bibitem{AN10}
L. Lv, Z. Ding, Q. Ni, and J. Chen, ``Secure MISO-NOMA transmission with artificial noise,'' \emph{IEEE Trans. Veh. Technol.}, vol. 67, no. 7, pp. 6700-6705, Jul. 2018.
\bibitem{AN11}
F. Wu, L. Yang, W. Wang, and Z. Kong, ``Secret precoding-aided spatial modulation,'' \emph{IEEE Commun. Lett.}, vol. 19, no. 9, pp. 1544-1547, Sep. 2015.
\bibitem{AN12}
F. Wu, R. Zhang, L. Yang, and W. Wang, ``Transmitter precoding-aided spatial modulation for secrecy communications,'' \emph{IEEE Trans. Veh. Technol.}, vol. 65, no. 1, pp. 467-471, Jan. 2016.
\bibitem{AN13}
H. Niu, X. Lei, Y. Xiao, Y. Li, and W. Xiang, ``Performance analysis and optimization of secure generalized spatial modulation,'' \emph{IEEE Trans. on Commun.}, vol. 68, no. 7, pp. 4451-4460, Jul. 2020.
\bibitem{AN14}
Y. Zhu, Y. Zhou, S. Patel, X. Chen, L. Pang, and Z. Xue, ``Artificial noise generated in MIMO scenario: Optimal power design,'' \emph{IEEE Signal Process. Lett.}, vol. 20, no. 10, pp. 964-967, Oct. 2013.
\bibitem{AN8}
N. Li, X. Tao, H. Wu, J. Xu, and Q. Cui, ``Large-system analysis of artificial-noise-assisted communication in the multiuser downlink: Ergodic secrecy sum rate and optimal power allocation,'' \emph{IEEE Trans. Veh. Technol.}, vol. 65, no. 9, pp. 7036-7050, Sep. 2016.
\bibitem{AN15}
S. Tsai and H. V. Poor, ``Power allocation for artificial-noise secure MIMO precoding systems,'' \emph{IEEE Trans. Signal Process.}, vol. 62, no. 13, pp. 3479-3493, Jul. 2014.
\bibitem{AN16}
Q. Li and W. Ma, ``Spatially selective artificial-noise aided transmit optimization for MISO multi-Eves secrecy rate maximization,'' \emph{IEEE Trans. Signal Process.}, vol. 61, no. 10, pp. 2704-2717, May 2013.
\bibitem{AN17}
F. Zhou, Z. Chu, H. Sun, R. Q. Hu, and L. Hanzo, ``Artificial noise aided secure cognitive beamforming for cooperative MISO-NOMA using SWIPT,'' \emph{IEEE J. Sel. Areas Commun.}, vol. 36, no. 4, pp. 918-931, Apr. 2018.
\bibitem{AN18}
H. Wang, T. Zheng, and X. Xia, ``Secure MISO wiretap channels with multiantenna passive eavesdropper: Artificial noise vs. artificial fast fading,'' \emph{IEEE Trans. Wirel. Commun.}, vol. 14, no. 1, pp. 94-106, Jan. 2015.
\bibitem{AN19}
H. Niu, X. Lei, Y. Xiao, D. Liu, Y. Li, and H. Zhang, ``Power minimization in artificial noise aided generalized spatial modulation,'' \emph{IEEE Commun. Lett.}, vol. 24, no. 5, pp. 961-965, May 2020.
\bibitem{MIMO1}
E. Shi et al., ``RIS-aided cell-free massive MIMO systems for 6G: Fundamentals, system design, and applications,'' \emph{Proc. IEEE}, vol. 112, no. 4, pp. 331-364, Apr. 2024.
\bibitem{AN20}
W. Wang, X. Chen, L. You, X. Yi, and X. Gao, ``Artificial noise assisted secure massive MIMO transmission exploiting statistical CSI,'' \emph{IEEE Commun. Lett.}, vol. 23, no. 12, pp. 2386-2389, Dec. 2019.
\bibitem{AN21}
W. Xu, B. Li, L. Tao, and W. Xiang, ``Artificial noise assisted secure transmission for uplink of massive MIMO systems,'' \emph{IEEE Trans. Veh. Technol.}, May 2021.
\bibitem{AN23}
K. Liu, J. Chen, and X. Lei, ``Artificial noise aided directional modulation: A transmitter architecture perspective,'' \emph{IEEE Trans. Commun.}, vol. 73, no. 5, pp. 3046-3060, May 2025.
\bibitem{AN24}
Y. Ding and V. F. Fusco, ``A vector approach for the analysis and synthesis of directional modulation transmitters,'' \emph{IEEE Trans. Antennas Propag.}, vol. 62, no. 1, pp. 361-370, Oct. 2013.
\bibitem{AN22}
J. Chen, Y. Xiao, K. Liu, Y. Zhong, X. Lei, and M. Xiao, ``Physical layer security for near-field communications via directional modulation,'' \emph{IEEE Trans. Veh. Technol.}, vol. 73, no. 8, pp. 12242-12246, Aug. 2024.
\bibitem{ANE2}
S. Liu, Y. Hong, and E. Viterbo, ``Artificial noise revisited,'' \emph{IEEE Trans. Inf. Theory.}, vol. 61, no. 7, pp. 3901-3911, Jul. 2015.
\bibitem{ANE1}
S. Liu, Y. Hong, and E. Viterbo, ``Artificial noise revisited: When Eve has more antennas than Alice,'' in \emph{Proc. 2014 Int. Conf. Signal Process. Commun. (SPCOM)}, Bangalore, India, Jul. 2014, pp. 1-5.
pp. 2996-3000.
\bibitem{ANE3}
H. Niu, Y. Xiao, X. Lei, and M. Xiao, ``Artificial noise elimination: From the perspective of eavesdroppers,'' \emph{IEEE Trans. Commun.}, vol. 70, no. 7, pp. 4745-4754, Jul. 2022.
\bibitem{ANE4}
H. Niu, Y. Xiao, X. Lei, G. Wang, M. Xiao, and S. Mumtaz, ``When the CSI from Alice to Bob is unavailable: What can Eve do to eliminate the artificial noise?,'' in \emph{Proc. IEEE 96th Veh. Technol. Conf. (VTC-Fall)}, London, United Kingdom, Sep. 2022, pp. 1-5.
\bibitem{Fisher1}
H. Niu, X. Lei, G. Wu, G. Wang, C. Yuen, and F. Adachi, ``Artificial noise elimination without the transmitter receiver link CSI,'' \emph{IEEE Trans. Veh. Technol.}, vol. 73, no. 9, pp. 13206-13218, Sep. 2024.
\bibitem{HC}
L. Liu, J. Liang, and K. Huang, ``Eavesdropping against artificial noise: Hyperplane clustering,'' in \emph{Proc. IEEE Third Int. Conf. Inf. Sci. Technol. (ICIST)}, Yangzhou, China, Feb. 2013, pp. 1571-1575.
\bibitem{PLI1}
M. M\'{e}dard and K. R. Duffy, ``Physical layer insecurity,'' in \emph{Proc. 57th Annu. Conf. Inf. Sci. Syst. (CISS)}, Baltimore, MD, USA, Mar. 2023, pp. 1-6.
\bibitem{PLI2}
K. R. Duffy and M. M\'{e}dard, ``Soft detection physical layer insecurity,'' in \emph{Proc. 2023 IEEE Global Commun. Conf.}, Kuala Lumpur, Malaysia, Dec. 2023, pp. 1747-1752.
\bibitem{SR1}
H. Shin and J. H. Lee, ``Closed-form formulas for ergodic capacity of MIMO Rayleigh fading channels,'' in \emph{Proc. IEEE Int.
Conf. Commun. (ICC)}, Anchorage, AK, USA, May 2003.
\bibitem{Wishart1}
A. Zanella, M. Chiani, and M. Z. Win, ``On the marginal distribution of the eigenvalues of wishart matrices,'' \emph{IEEE Trans. Commun.}, vol. 57, no. 4, pp. 1050-1060, Apr. 2009.
\bibitem{CC1}
E. Telatar, ``Capacity of multi-antenna Gaussian channels,'' \emph{Eur. Trans. Telecommun.}, vol. 10, no. 6, pp. 585-595, 1999.

\end{thebibliography}
\end{document}